\documentclass[10pt,a4paper,fleqn]{article}

\usepackage[labelfont=bf]{caption}

\usepackage{lscape,exscale,amsthm,enumerate}
\usepackage[intlimits]{amsmath}
\usepackage{rawfonts}
\usepackage{latexsym}
\usepackage[cp850]{inputenc}
\usepackage{epsfig}
\usepackage{bbm}
\usepackage{float}
\usepackage{mathtools,dsfont}

\RequirePackage[OT1]{fontenc}
\RequirePackage{amsthm,amsmath}
\RequirePackage[numbers]{natbib}
\RequirePackage[colorlinks,citecolor=blue,urlcolor=blue]{hyperref}
\newcommand{\bol}[1]{\mbox{\boldmath$#1$}}

\newcommand{\bSigma}{\bol{\Sigma}}
\newcommand{\bOmega}{\bol{\Omega}}
\newcommand{\hbSigma}{\hat{\bol{\Sigma}}}

\newcommand{\bmu}{\bol{\mu}}
\newcommand{\blambda}{\bol{\lambda}}
\newcommand{\hbmu}{\hat{\bol{\mu}}}
\newcommand{\tbmu}{\tilde{\bol{\mu}}}
\newcommand{\bet}{\bol{\eta}}
\newcommand{\btheta}{\bol{\theta}}
\newcommand{\bb}{\mathbf{b}}

\newcommand{\bx}{\mathbf{x}}
\newcommand{\bQ}{\mathbf{Q}}
\newcommand{\hbQ}{\hat{\mathbf{Q}}}

\newcommand{\by}{\mathbf{y}}

\newcommand{\bt}{\mathbf{t}}

\newcommand{\bP}{\mathbf{P}}
\newcommand{\bM}{\mathbf{M}}
\newcommand{\bH}{\mathbf{H}}
\newcommand{\bL}{\mathbf{L}}
\newcommand{\bl}{\mathbf{l}}

\newcommand{\hbet}{\bol{\hat{\eta}}}

\newcommand{\bz}{\mathbf{z}}
\newcommand{\bd}{\mathbf{d}}
\newcommand{\bB}{\mathbf{B}}

\newcommand{\bv}{\mathbf{v}}
\newcommand{\bw}{\mathbf{w}}
\newcommand{\hbw}{\mathbf{\hat{w}}}

\newcommand{\bi}{\mathbf{1}}

\newcommand{\bI}{\mathbf{I}}

\newcommand{\bA}{\bol{A}}

\newcommand{\bF}{\mathbf{F}}

\newcommand{\bD}{\mathbf{D}}

\newcommand{\bG}{\mathbf{G}}
\newcommand{\bu}{\mathbf{u}}

\newcommand{\hbtheta}{\bol{\hat{\theta}}}

\usepackage{stackrel}

\numberwithin{equation}{section}
\theoremstyle{plain}
\newtheorem{theorem}{Theorem}[section]

\newtheorem{lemma}{Lemma}[section]
\newtheorem{corollary}{Corollary}[section]

\usepackage{color}

\usepackage{ulem}

\usepackage{natbib}
\providecommand{\keywords}[1]{\textbf{\textit{Keywords---}} #1}

\newtheorem{innercustomthm}{Theorem}
\newenvironment{customthm}[1]
  {\renewcommand\theinnercustomthm{#1}\innercustomthm}
  {\endinnercustomthm}

\newtheorem{innercustomlem}{Lemma}
\newenvironment{customlem}[1]
  {\renewcommand\theinnercustomlem{#1}\innercustomlem}
  {\endinnercustomlem}
  
\newtheorem{innercustomcor}{Corollary}
\newenvironment{customcor}[1]
  {\renewcommand\theinnercustomcor{#1}\innercustomcor}
  {\endinnercustomcor}
  
\begin{document}

\begin{center}
  \vspace*{1cm} \noindent {\bf \large Corrigendum to\\
  "Sampling Distributions of Optimal Portfolio Weights and Characteristics in Small and Large Dimensions"}  \\

\vspace{1cm} \noindent {\sc  Taras Bodnar$^{a}$, Holger Dette$^{b}$,
Nestor Parolya$^{c}$ and Erik Thors\'{e}n$^{a}$}\\
 \vspace{0.5cm}
 {\it \footnotesize  $^a$
 Department of Mathematics, Stockholm University, Stockholm, Sweden
 } \\
 {\it \footnotesize  $^b$
 Department of Mathematics, Ruhr University Bochum, D-44870 Bochum, Germany
 }\\
 {\it \footnotesize  $^c$
 Delft Institute of Applied Mathematics, Delft University of Technology, Delft, The Netherlands }

\end{center}
\vspace{0.5cm}
  We would like to thank Raymond Kan, Nathan Lassance and Xiaolu Wang for drawing our attention to the errors present in the paper ``Sampling Distributions of Optimal Portfolio Weights and Characteristics in Small and Large Dimensions''. The correct statements of the theorems are presented below, in which we used the notations from the original paper. For further important stochastic representations of estimated optimal portfolio weights and related asymptotic theory we refer to \cite{KanLassanceWang2023}.

Although the proof of Lemma A.1 in \cite{bodnar2022sampling} is correct, a typographical error is present in its formulation. The corrected version of this lemma is given below:

\begin{customlem}{A.1}\label{th1}
Under the conditions of Theorem 2.1 in \cite{bodnar2022sampling}, the distribution of $(\hat{V}_{GMV},\hat{\btheta}^\top,\tilde{R}_{GMV},\tilde{s},\tilde{\bet}^\top)^\top$ is determined by
\begin{enumerate}[(i)]
\item $\hat{V}_{GMV}$ is independent of $(\hat{\btheta}^\top,\tilde{R}_{GMV},\tilde{s},\tilde{\bet}^\top)^\top$;
\item $(n-1)\frac{\hat{V}_{GMV}}{V_{GMV}}\sim \chi^2_{n-p}$;
\item $\left(\begin{array}{c} \hat{\btheta} \\ \tilde{R}_{GMV}\end{array}\right) \sim
    t_{k+1}\left( n-p+1,\left(\begin{array}{c} \btheta \\ \breve{R}_{GMV}\end{array}\right),
     \frac{V_{GMV}}{n-p+1} \breve{\bG} \right)$, with~ $\breve{\bG}=\begin{pmatrix} \bL \bQ \bL^\top & \bL \bQ \tbmu \\ \tbmu^\top \bQ \bL^\top & \tbmu^\top \bQ \tbmu \end{pmatrix}
    =\begin{pmatrix} \bL \bQ \bL^\top & \breve{s} \breve{\bet} \\ \breve{s} \breve{\bet}^\top & \breve{s} \end{pmatrix}
   $
\item $\tilde{s}$ and $\tilde{\bet}$ are conditionally independent given $\hat{\btheta}$ and $\tilde{R}_{GMV}$

\item $(n-1)\frac{\breve{s}}{\tilde{s}}\left(1 + \frac{(\tilde{R}_{GMV} - \breve{R}_{GMV})^2}{V_{GMV}\breve{s}}\right) \sim \chi^2_{n-p+2}$;
\item
$\tilde{\bet}|\hat{\btheta}^\top, \tilde{R}_{GMV} \sim
t_{k}\left(n-p+3,\breve{\bet}+ \mathbf{h},
\frac{(n-p+3)^{-1}\tilde{\bF}}{\breve{s}\left(1+\frac{(\tilde{R}_{GMV} - \breve{R}_{GMV})^2}{V_{GMV}\breve{s}}\right)^2} \right), $
where {\small
\begin{eqnarray*}
\mathbf{h}&=&\left(1 + \frac{(\tilde{R}_{GMV} - \breve{R}_{GMV})^2}{V_{GMV}\breve{s}}\right)^{-1}\frac{(\hat{\btheta}-\btheta-\breve{\bet}(\tilde{R}_{GMV} - \breve{R}_{GMV}))(\tilde{R}_{GMV} - \breve{R}_{GMV})}{V_{GMV}\breve{s}}\\
\tilde{\bF}&=&\left(\bL \bQ \bL^\top-\breve{s}\breve{\bet}\breve{\bet}^\top\right)\left(1+\frac{(\tilde{R}_{GMV} - \breve{R}_{GMV})^2}{V_{GMV}\breve{s}}\right)\\
&+&\frac{1}{V_{GMV}}\left(\hbtheta - \btheta-\breve{\bet}(\tilde{R}_{GMV} - \breve{R}_{GMV})\right)\left(\hbtheta - \btheta-\breve{\bet}(\tilde{R}_{GMV} - \breve{R}_{GMV})\right)^\top.
\end{eqnarray*}
}
\end{enumerate}
\end{customlem}

The new formulation of Lemma A.1 impacts the statements of Theorem 2.1 in \cite{bodnar2022sampling}, which is corrected to

\begin{customthm}{2.1}\label{th2}
Let $\bx_1,\bx_2,...,\bx_n$ be independent and normally distributed with mean vector $\bmu$ and covariance matrix $\bSigma$, i.e. $\bx_i\sim \mathcal{N}_p(\bmu,\bSigma)$ for $i=1,...,n$ with $n>p$. Define $\bM=(\bL^\top,\tbmu,\bi)^\top$ and assume that $rank(\bM)=k+2$. Let $\bSigma$ be positive definite. Then, a joint stochastic representation of $\hat{V}_{GMV}$, $\hat{R}_{GMV}$, $\hat{\btheta}$, $\hat{s}$, and $\hat{\bet}$ is given by
\begin{enumerate}[(i)]
\item \label{th2_VG} $\hat{V}_{GMV}\stackrel{d}{=}\frac{V_{GMV}}{n-1}\xi_1$;
\item \label{th2_RG}
$\hat{R}_{GMV}\stackrel{d}{=}R_{GMV}+\sqrt{V_{GMV}}\left(\frac{z_1}{\sqrt{n}}+\sqrt{f}\frac{t_1}{\sqrt{n-p+1}}\right)$;
\item 
\begin{eqnarray*}
\hat{\btheta}&\stackrel{d}{=}&\btheta+\sqrt{V_{GMV}}
\Bigg(\frac{s\bet+\bz_2/\sqrt{n}}{\sqrt{f}}\frac{t_1}{\sqrt{n-p+1}}\nonumber \\
&+&\left(\bL \bQ \bL^\top-\frac{\left(s\bet+\bz_2/\sqrt{n}\right)\left(s\bet+\bz_2/\sqrt{n}\right)^\top}
{f}\right)^{1/2}
\sqrt{1+\frac{t_1^2}{n-p+1}}\frac{\bt_2}{\sqrt{n-p+2}}\Bigg);\nonumber
\end{eqnarray*}
\item \label{th2_s}$\hat{s}\stackrel{d}{=}(n-1)\left(1+\frac{t_1^2}{n-p+1}\right)\frac{f}{\xi_2}$ with
\begin{equation*}\label{th2_f}
f=\frac{\xi_3}{n}+ \left(s\bet+ \frac{\bz_2}{\sqrt{n}}\right)^\top\left(\bL \bQ \bL^\top\right)^{-1}\left( s\bet+ \frac{\bz_2}{\sqrt{n}}\right);
\end{equation*}
\item \label{th2_bet} \begin{eqnarray*}
\hat{\bet}&\stackrel{d}{=}&\frac{s\bet+\bz_2/\sqrt{n}}{f}
+\frac{1}{\sqrt{f}\sqrt{1+\frac{t_1^2}{n-p+1}}}\left(\bL \bQ \bL^\top-\frac{\left(s\bet+\bz_2/\sqrt{n}\right)\left(s\bet+\bz_2/\sqrt{n}\right)^\top}
{f}\right)^{1/2} \nonumber \\
&\times& \left(\frac{\bt_2}{\sqrt{n-p+2}}\frac{t_1}{\sqrt{n-p+1}}
+ \left(\bI_k+\frac{\bt_2\bt_2^\top}{n-p+2}\right)^{1/2}\frac{\bt_3}{\sqrt{n-p+3}} \right)\nonumber
\end{eqnarray*}
\end{enumerate}
where $\xi_1\sim \chi^2_{n-p}$, $\xi_2\sim \chi^2_{n-p+2}$, $\xi_3\sim \chi^2_{p-k-1;n\bmu^\top\bA \bmu}$, $z_1\sim \mathcal{N}(0,1)$, $\bz_2\sim \mathcal{N}_k(\mathbf{0},\bL \bQ \bL^\top)$, $t_1 \sim t(n-p+1)$, $\bt_2\sim t_k(n-p+2)$, and $\bt_3 \sim t_k(n-p+3)$ are mutually independent with
$\bA=\bQ- \bQ \bL^\top \left(\bL \bQ \bL^\top\right)^{-1}\bL \bQ$.
\end{customthm}

\begin{proof}[Proof of Theorem 2.1:]
The new formulation of Lemma \ref{th1} has an influence on the derivation of the stochastic representation of $\hat{\bet}$ only, while the remaining parts of the proof are correct. The derivation of the new stochastic representation of $\hat{\bet}$ is given by
{\footnotesize
\begin{eqnarray*}
\hat{\bet}&\stackrel{d}{=}&\frac{\bL\bQ \hbmu}{\hbmu^\top \bQ \hbmu}
+\frac{\sqrt{1+\frac{t_1^2}{n-p+1}}\left(\bL \bQ \bL^\top-\frac{\bL\bQ \hbmu\hbmu^\top \bQ \bL^\top}{\hbmu^\top \bQ \hbmu}\right)^{1/2}\frac{\bt_2}{\sqrt{n-p+2}} \frac{1}{\sqrt{\hbmu^\top \bQ \hbmu}}\frac{t_1}{\sqrt{n-p+1}}
                            }{1+\frac{t_1^2}{n-p+1}}\nonumber\\
  &+&\frac{1}{\sqrt{\hbmu^\top \bQ \hbmu}} \frac{1}{1+\frac{t_1^2}{n-p+1}}
 \Bigg(\left(\bL \bQ \bL^\top-\frac{\bL\bQ \hbmu\hbmu^\top \bQ \bL^\top}{\hbmu^\top \bQ \hbmu}\right)\left(1+\frac{t_1^2}{n-p+1}\right)\nonumber\\
&+&\frac{1}{V_{GMV}}\left(1+\frac{t_1^2}{n-p+1}\right)\frac{V_{GMV}}{n-p+2}
\left(\bL \bQ \bL^\top-\frac{\bL\bQ \hbmu\hbmu^\top \bQ \bL^\top}{\hbmu^\top \bQ \hbmu}\right)^{1/2}\bt_2\bt_2^\top \left(\left(\bL \bQ \bL^\top-\frac{\bL\bQ \hbmu\hbmu^\top \bQ \bL^\top}{\hbmu^\top \bQ \hbmu}\right)^{1/2}\right)^\top
\Bigg)^{1/2}\frac{\bt_3}{\sqrt{n-p+3}}\nonumber\\
&=&\frac{\bL\bQ \hbmu}{\hbmu^\top \bQ \hbmu}
+\frac{1}{\sqrt{\hbmu^\top \bQ \hbmu}}\frac{1}{\sqrt{1+\frac{t_1^2}{n-p+1}}}\left(\bL \bQ \bL^\top-\frac{\bL\bQ \hbmu\hbmu^\top \bQ \bL^\top}
    {\hbmu^\top \bQ \hbmu}\right)^{1/2}\nonumber\\
  &\times&\left(\frac{\bt_2}{\sqrt{n-p+2}}\frac{t_1}{\sqrt{n-p+1}}
+ \left(\bI_k+\frac{\bt_2\bt_2^\top}{n-p+2}\right)^{1/2}\frac{\bt_3}{\sqrt{n-p+3}} \right)\label{ap_th2_eq5}
\end{eqnarray*}
}
where $\bt_3 \sim t_k(n-p+3)$ and is independent of $t_1$ and $\bt_2$. Moreover, due to Lemma \ref{th1}.i and \ref{th1}.iv we get that $\xi_1$, $\xi_2$, $t_1$, $\bt_2$, and $\bt_3$ are mutually independent.
\end{proof}

The new statement of Theorem 2.1 leads to the corrected version of Theorem 3.1 given by

\begin{customthm}{3.1}\label{th3}
  Under the conditions of Theorem 2.1 in \cite{bodnar2022sampling}, it holds that
\begin{flalign*}
\bL \hat{\bw}_g &\stackrel{d}{=} \btheta + \left( \sqrt{\frac{V_{GMV}}{f}} \frac{t_1}{\sqrt{n-p+1}}+
\frac{g(\hat{R}_{GMV}, \hat{V}_{GMV}, \hat{s})}{f} \right) \left(s\bet+\bz_2/\sqrt{n} \right)\nonumber \\ 
&+\frac{1}{\sqrt{1+\frac{t_1^2}{n-p+1}}} \left( \bL \bQ \bL^\top-\frac{\left(s\bet+\bz_2/\sqrt{n}\right)\left(s\bet+\bz_2/\sqrt{n}\right)^\top}{f} \right)^{1/2}\nonumber \\ 
&\times\Bigg[\Bigg(\sqrt{V_{GMV}} \left(1+\frac{t_1^2}{n-p+1} \right) + \frac{g(\hat{R}_{GMV}, \hat{V}_{GMV}, \hat{s})}{\sqrt{f}} \frac{t_1}{\sqrt{n-p+1}}
 \Bigg)\frac{\bt_2}{\sqrt{n-p+2}}\nonumber\\
 &+ \frac{g(\hat{R}_{GMV}, \hat{V}_{GMV}, \hat{s})}{\sqrt{f}}  \left( \bI_k + \frac{\bt_2 \bt_2^\top}{n-p+2}\right)^{1/2} \frac{\bt_3}{\sqrt{n-p+3}}\Bigg]\label{eqn:th3_samp}
\end{flalign*}
where the joint stochastic representation of $\hat{V}_{GMV}$, $\hat{R}_{GMV}$ and $\hat{s}$ is given in \eqref{th2_VG}-\eqref{th2_bet} of Theorem 2.1.
\end{customthm}

The corrected version of Corollary 4.1 is the following one

\begin{customcor}{4.1}\label{cor1}
 Under the conditions of Theorem~\ref{th2} and Assumption (\textbf{A1}) in \cite{bodnar2022sampling}, it holds that
\begin{equation*}
 \sqrt{n-p}\begin{pmatrix}
 \hat{V}_{GMV} - \frac{1-p/n}{1-1/n} V_{GMV} \\
 \hat{R}_{GMV}- R_{GMV} \\
 \hat{\btheta} - \btheta \\
 \hat{s}-\frac{(s+p/n)(1-1/n)}{1-p/n+2/n} \\
 \hat{\bet} - \frac{s}{s+p/n}\bet
 \end{pmatrix} \rightarrow N_{2k+3}\left(\mathbf{0},\mathbf{\Xi} \right)
\end{equation*}
with
\begin{equation*}
\mathbf{\Xi}=
 \begin{pmatrix}
 2 V^2_{GMV} (1-c)^2 & 0 & 0 & 0 & 0 \\
 0 & V_{GMV}(1+s) &  V_{GMV} s \bet^\top & 0 & 0  \\
 0 & V_{GMV} s \bet &   V_{GMV} \bL \bQ \bL^\top & 0 & 0  \\
 0 & 0 & 0 & \Xi_{s,s} & \boldsymbol{\Xi}^\top_{s, \bet}  \\
 0 & 0 & 0 & \boldsymbol{\Xi}_{s, \bet}  &  \boldsymbol{\Xi}_{\bet,\bet}
 \end{pmatrix}
\end{equation*}
for $p/n \rightarrow c\in [0,1)$ as $n \rightarrow \infty$ where
 \begin{eqnarray*}
\Xi_{s,s}&=& \frac{2(c+2s)}{(1-c)} + 2\frac{(s+c)^2}{(1-c)^2}, \label{eqn:s_asymp}\\
\boldsymbol{\Xi}_{\bet,\bet}&=& \frac{s+1}{(s+c)^2}\bL\bQ\bL^\top - \frac{s^2(2c(1-c)+(s+c)^2)}{(s+c)^4}\bet \bet^\top, \label{eqn:bet_asymp} \\
\boldsymbol{\Xi}_{s, \bet} &=&-\frac{2s^2}{(s+c)^2}\bet. \nonumber
\end{eqnarray*}
\end{customcor}

\begin{proof}[Proof of Corollary 4.1:]
Only, the expression of $\boldsymbol{\Xi}_{s, \bet}$ is changed. The new formula is obtained from Theorem 4.1 in \cite{bodnar2022sampling} and it is given by
\begin{eqnarray*}
\boldsymbol{\Xi}_{s, \bet} &=&-\frac{2s}{(s+c)^2}(c+2\bmu^\top \bA \bmu)\bet+\frac{2s}{s+c}\left(\bL\bQ\bL^\top-\frac{2s^2}{s+c}\bet\bet^\top\right)(\bL\bQ\bL^\top)^{-1} \bet=-\frac{2s^2}{(s+c)^2}\bet.
\end{eqnarray*}
\end{proof}

The improved statement of Theorem 4.2 is given by
\begin{customthm}{4.2}\label{th6}
Let $g(.,.,.)$ be differentiable with first order continuous derivatives. Then, under the conditions of Theorem~\ref{th2} and Assumption (\textbf{A1}) in \cite{bodnar2022sampling}, we get
\begin{eqnarray*}
&& \sqrt{n-p}\left(\bL\hat{\bw}_g - \left(\btheta + \frac{s g\left(\blambda\right)}{s+p/n}\bet \right) \right)
 \stackrel{d}{\rightarrow} N_k(\mathbf{0}, \bOmega_{\bL,g})
\end{eqnarray*}
for $p/n \rightarrow c\in [0,1)$ as $n \rightarrow \infty$ with
\begin{eqnarray*}
&&\bOmega_{\bL,g} =  \bigg(\frac{(s+1) g\left(\blambda\right)^2}{(s+c)^2} + V_{GMV}  \bigg) \bL \bQ \bL^\top
+ s^2\Bigg\{
	2\frac{(1-c)^2 V^2_{GMV}}{(s+c)^2}  g_2\left(\blambda\right)^2\nonumber \\
&	+&  \left(\frac{g_3\left(\blambda\right)}{1-c} - \frac{ g\left(\blambda\right)}{s+c} \right)^2 \frac{2(1-c)(c+2s)}{(s+c)^2} + \frac{4(1-c)}{(s+c)^2}
 g\left(\blambda\right) \left(\frac{g_3\left(\blambda\right)}{1-c}  - \frac{g\left(\blambda\right)}{s+c} \right)
 \nonumber \\
&	+ & \frac{ V_{GMV}(s+1) }{(s+c)^2} g_1\left(\blambda\right)^2 +  2\frac{V_{GMV}}{(s+c)} g_1\left(\blambda\right) + \frac{2}{(1-c)^2} g_3\left(\blambda\right)^2 -\frac{ g\left(\blambda\right)^2}{(s+c)^2} \Bigg\} \bet \bet^\top. \label{Omega_g}
\end{eqnarray*}
\end{customthm}

As a special case of Theorem 4.2 we get the formula of the asymptotic covariance matrix expressed as
 \begin{flalign*}
\bOmega_{\bL,EU}&=\bigg( \gamma^{-2} \frac{s+1}{(1-c)^2} + V_{GMV}  \bigg) \bL \bQ \bL^\top +\frac{\gamma^{-2}s^2}{(1-c)^2}\bet \bet^\top.
\end{flalign*}

The correct version of Theorem 4.4 is given by
\begin{customthm}{4.4}\label{th8}
Let $\blambda=(R_{GMV},V_{GMV},s)^\top$. Then, under the conditions of Theorems 4.2 and 4.3 in \cite{bodnar2022sampling}, it holds that
\begin{enumerate}[(a)]
\item $ \sqrt{n-p}\left(\bL\hat{\bw}_{g;c} - \bL\bw_g \right) \stackrel{d}{\rightarrow} N_k(\mathbf{0}, \bOmega_{\bL,g,c})
$
for $p/n \rightarrow c\in[0,1)$ as $n \rightarrow \infty$ with
{\small
\begin{eqnarray*}\label{Omega_gc}
\bOmega_{\bL,g,c}&=& \bigg( \frac{s+1}{s^2} g\left(\blambda_0\right)^2 + V_{GMV}  \bigg) \bL \bQ \bL^\top\nonumber\\
  &+&\Bigg\{2 V^2_{GMV} g_2\left(\blambda_0\right)^2+2 (s^2+2s+c) \left(g_3\left(\blambda_0\right)-\frac{g\left(\blambda_0\right)}{s}\right)^2
  +4 (s+1) g\left(\blambda_0\right)\left(g_3\left(\blambda_0\right)-\frac{g\left(\blambda_0\right)}{s}\right)
 \nonumber\\
& +& (s+1)V_{GMV} g_1\left(\blambda_0\right)^2 + 2sV_{GMV} g_1\left(\blambda_0\right)+  g\left(\blambda_0\right)^2\Bigg\} \bet \bet^\top
 \nonumber;
\end{eqnarray*}
}
\item for $p/n \rightarrow c\in[0,1)$ as $n \rightarrow \infty$ holds
\begin{equation*}
 \sqrt{n-p}\begin{pmatrix}
 \hat{h}_{g,1,c}- h_{g,1}\left(\blambda_0\right) \\
 \vdots \\
 \hat{h}_{g,q,c}- h_{g,q}\left(\blambda_0\right) \\
 \end{pmatrix} \rightarrow N_q\left(\mathbf{0},\mathbf{\Xi}_{h,c} \right)~~\text{with~~ $\mathbf{\Xi}_{h,c}=\left(\Xi_{h,c;ij}\right)_{i,j=1,...,q}$ }
\end{equation*}
where
\begin{equation*}
\Xi_{h,c;ij}=V_{GMV}(1+s) h_{g,i;1}\left(\blambda_0\right)h_{g,j;1}\left(\blambda_0\right)
+2 V^2_{GMV} h_{g,i;2}\left(\blambda_0\right)h_{g,j;2}\left(\blambda_0\right)
+\left(2s^2+4s+2c\right)h_{g,i;3}\left(\blambda_0\right)h_{g,j;3}\left(\blambda_0\right).\label{Xi_hc}
\end{equation*}
\end{enumerate}
\end{customthm}

Finally, a consistent estimator for the covariance matrix of the estimated weights of the EU portfolio is given by:
 \begin{flalign*}
\hat{\bOmega}_{\bL,EU,c}&=\bigg( \gamma^{-2}(\hat{s}_c+1) + \hat{V}_{GMV;c} \bigg) (1-c_n)\bL \hat{\bQ} \bL^\top + \gamma^{-2}\hat{s}_c^2  \hat{\bet}_c \hat{\bet}_c^\top,
\end{flalign*}
where $c_n=p/n$ and $\hat{s}_c= \frac{n-p}{n}\left(\hat{s}-\frac{p}{n-p}\right)$.

\newpage

\begin{center}
  \vspace*{1cm} \noindent {\bf \large Sampling Distributions of Optimal Portfolio Weights and Characteristics in Small and Large Dimensions}  \\

\vspace{1cm} \noindent {\sc  Taras Bodnar$^{a}$, Holger Dette$^{b}$,
Nestor Parolya$^{c}$ and Erik Thors\'{e}n$^{a}$}\\
 \vspace{0.5cm}
 {\it \footnotesize  $^a$
 Department of Mathematics, Stockholm University, Stockholm, Sweden
 } \\
 {\it \footnotesize  $^b$
 Department of Mathematics, Ruhr University Bochum, D-44870 Bochum, Germany
 }\\
 {\it \footnotesize  $^c$
 Delft Institute of Applied Mathematics, Delft University of Technology, Delft, The Netherlands }

\end{center}
\vspace{0.5cm}






\begin{abstract}
Optimal portfolio selection problems are determined by the (unknown) parameters of the data generating process. If an investor wants to realise the position suggested by the optimal portfolios, he/she needs to estimate the unknown parameters and to account for the parameter uncertainty in the decision process. Most often, the parameters of interest are the population mean vector and the population covariance matrix of the asset return distribution. In this paper, we characterise the exact sampling distribution of the estimated optimal portfolio weights and their characteristics. This is done by deriving their sampling distribution by its stochastic representation. This approach possesses several advantages, {e.g.} (i) it determines    the sampling distribution of the estimated optimal portfolio weights by expressions, which could be used to draw samples from this distribution efficiently; (ii) the application of the derived stochastic representation provides an easy way to obtain the asymptotic approximation of the sampling distribution. The later property is used to show that the high-dimensional asymptotic distribution of optimal portfolio weights is a multivariate normal and to determine its parameters. Moreover, a consistent estimator of optimal portfolio weights and their characteristics is derived under the high-dimensional settings. Via an extensive simulation study, we investigate the finite-sample performance of the derived asymptotic approximation and study its robustness to the violation of the model assumptions used in the derivation of the theoretical results.
\end{abstract}

\section{Introduction}
The solution to the optimal portfolio selection problems are determined by the parameters of the data generating process. In many cases, the optimal portfolio weights and their characteristics, like the portfolio mean, the portfolio variance, the value-at-risk (VaR), the conditional VaR (CVaR), etc., can be computed by using only the mean vector and the covariance matrix of the asset return distribution. More precisely, these relationships are summarized by the following five quantities:
\begin{equation}\label{5quant}
V_{GMV}=\frac{1}{\bi^\top \bSigma^{-1}\bi},~~
\bw_{GMV}=\frac{\bSigma^{-1}\bi}{\bi^\top \bSigma^{-1}\bi},~~
R_{GMV}=\frac{\bmu^\top\bSigma^{-1}\bi}{\bi^\top \bSigma^{-1}\bi},~~
s=\bmu^\top \bQ \bmu,~~
\bv=\frac{\bQ \bmu}{\bmu^\top \bQ \bmu},
\end{equation}
where $\bmu=E(\bx)$ and $\bSigma=Var(\bx)$ are the mean vector and the covariance matrix of the $p$-dimensional asset return vector $\bx$ and
\begin{equation}\label{bQ}
\bQ=\bSigma^{-1}-\frac{\bSigma^{-1}\bi\bi^\top \bSigma^{-1}}{\bi^\top \bSigma^{-1}\bi}.
\end{equation}

The five quantities in \eqref{5quant} have an interesting financial interpretation. The components of the $p$-dimensional vector $\bw_{GMV}$ define  the {weights}  of the global minimum variance (GMV) portfolio, i.e. of the portfolio with the smallest variance, while $R_{GMV}$ and $V_{GMV}$ are the expected return and the variance of the GMV portfolio. The quantity $s$ is the slope parameter of the efficient frontier, the set of all optimal portfolios following Markowitz's approach. This parameter, together with $R_{GMV}$ and $V_{GMV}$, fully determine the location and the shape of the efficient frontier, which is a parabola in the mean-variance space. Finally, the components of the
 $p$-dimensional vector $\bv$ define the {weights}  of the so-called self-financing portfolio (cf. \citet{korkie2002mean}), i.e. the sum of its weights is equal to zero, that is $\bi^\top \bv=0$.

The five quantities in \eqref{5quant} determine the structure of many optimal portfolios, like the GMV portfolio, the mean-variance (MV) portfolio, the expected maximum exponential utility (EU) portfolio, the tangency (T) portfolio, the optimal portfolio that maximizes the Sharpe ratio (SR), the minimum VaR (MVaR) portfolio, and the minimum CVaR (MCVaR) portfolio, maximum value-of-return (MVoR) portfolio, maximum conditional value-of-return (MCVoR) portfolio, among others (see, e.g., \citet{markowitz1952portfolio}, \citet{ingersoll1987theory}, \citet{korkie1981}, \citet{alexander2002economic}, \citet{alexander2004comparison}, \citet{okhrin2006distributional}, \citet{kan2007optimal}, \citet{frahm2010dominating}, \citet{bodnar2012minimum}, \citet{adcock2015statistical}, \citet{woodgate2015much}, \citet{bodnar2018estimation}, \citet{BodnarLindholmThorsenTyrcha2018},  \citet{simaan2018estimation}, \citet{bodnar2019okhrin}, \citet{bodnardmytrivparolyaschmid2019}). On the other hand, the quantities \eqref{5quant} cannot be directly used to compute the weights and the characteristics of these portfolios, since both $\bmu$ and $\bSigma$ are unobservable parameters in practice. As a result, an investor determines the optimal portfolios by replacing $\bmu$ and $\bSigma$ in \eqref{5quant} with the corresponding sample estimators given by
\begin{equation}\label{sample_mean_var}
\hbmu=\frac{1}{n}\sum_{i=1}^{n} \bx_i \quad \text{and} \quad \hbSigma=\frac{1}{n-1} \sum_{i=1}^{n} (\bx_i-\hbmu)(\bx_i-\hbmu)^\top,
\end{equation}
given a sample of asset returns $\bx_1,\bx_2,...,\bx_n$. This approach leads to the sample or the so-called plug-in estimators of the optimal portfolios, which are based on the corresponding sample estimators of \eqref{5quant}, expressed as
\begin{equation}\label{sample_5quant}
\hat{V}_{GMV}=\frac{1}{\bi^\top \hbSigma^{-1}\bi}, ~~
\hat{\bw}_{GMV}=\frac{\hbSigma^{-1}\bi}{\bi^\top \hbSigma^{-1}\bi},~~
\hat{R}_{GMV}=\frac{\hbmu^\top\hbSigma^{-1}\bi}{\bi^\top \hbSigma^{-1}\bi},~~
\hat{s}=\hbmu^\top \hbQ \hbmu,~~
\hat{\bv}=\frac{\hbQ \hbmu}{\hbmu^\top \hbQ \hbmu},
\end{equation}
with
\begin{equation}\label{sample_bQ}
\hbQ=\hbSigma^{-1}-\frac{\hbSigma^{-1}\bi\bi^\top \hbSigma^{-1}}{\bi^\top \hbSigma^{-1}\bi}
\end{equation}
as well as to the sample (plug-in) estimators of the optimal portfolio weights.

The notion of the sampling distribution in portfolio allocation has recently been given large attention. Investors and researchers realize that the uncertainty, introduced by using historical data, needs to be integrated into the optimal portfolio decision process as well as  properly assessed. The sampling distribution of the mean-variance portfolio was investigated  as early as \citet{korkie1981}, \citet{britten1999sampling}, \citet{okhrin2006distributional}, where the distributions of estimated optimal portfolio weights were derived under the assumption of an independent sample of asset returns taken from a multivariate normal distribution. Moreover, both the asymptotic and finite-sample distributions of the estimated efficient frontier, the set of all mean-variance optimal portfolios, were obtained by \citet{jobson1991confidence}, \citet{bodnar2008estimation}, \citet{kan2008distribution}, and \citet{bodnar2009econometrical}, among others, while \citet{siegel2007performance} and \citet{bodnar2010unbiased} presented its improved estimators and proposed a test of its existence. Some of these results were later extended to the high-dimensional setting in \citet{frahm2010dominating}, \citet{glombeck}, \citet{bodnar2018estimation}, \citet{bodnar2019okhrin}, whereas several limiting results related to the estimation of optimal portfolios under high-dimensional settings are present in \citet{ao2019approaching}, \citet{kan2019sample}, \citet{cai2020high}, \citet{ding2020high}, \citet{bodnardmytrivokhrinparolyaschmid2020}, among others.

The sample mean vector and the sample covariance matrix given by \eqref{sample_mean_var} have been used extensively in previous research (see, e.g., \citet{britten1999sampling}, \citet{memmel2006estimating}, \citet{okhrin2008estimation}) for estimating the asset return vector and its covariance matrix. These estimators appear to be consistent and the corresponding
estimated optimal portfolios  have desirable asymptotic properties when the portfolio dimension is considerably smaller than the sample size. However, they can no longer be used when a high-dimensional portfolio is constructed, due to their performance when the portfolio dimension is comparable to the sample size. One of the issues lies in that the quantities \eqref{sample_5quant} depend on the inverse covariance matrix. Its sample counterpart is not a consistent estimator in the high-dimensional settings (see, e.g., \citet{BodnarGuptaParolya2016}). To cope with these limitations, a number of improved estimators have been considered in the literature (cf., \citet{efron1976}, \citet{JagannathanMa2003}, \citet{golosnoy2007multivariate}, \citet{frahm2010dominating}, \citet{DeMiguel2009b}, \citet{rubioetal2012}, \citet{yao_zheng_bai_2015}).

We contribute to the existent literature by deriving the joint sampling distribution of the estimated five quantities in \eqref{sample_5quant}, which solely determine the structure of optimal portfolios. These results are then used to establish a unified approach for characterizing the sampling distributions of the estimated weights and the corresponding estimated characteristics of optimal portfolios. The goal is achieved by presenting the joint distribution of $(\hat{V}_{GMV},\hat{\bw}_{GMV}^\top,\hat{R}_{GMV},\hat{s},\hat{\bv}^\top)^\top$ in terms of a very useful stochastic representation. A stochastic representation is a computationally efficient tool in statistics and econometrics to characterize the distribution of a random variable/vector, which is widely used in both conventional and Bayesian statistics. While it plays a special rule in the theory of elliptical distributions (c.f., \citet{gupta2013elliptically}), the stochastic representation is also a very popular method to generate random variables/vectors in computational statistics (see, e.g., \citet{givens2012computational}). The applications of stochastic representations in the determination of the posterior distributions of estimated optimal portfolios can be found in \citet{bodnar2017bayesian} and \citet{bauder2019bayesian}. Finally, \citet{zellner2010direct}, among others, argued that the direct Monte Carlo approach based on stochastic representations is a computationally efficient method to calculate Bayesian estimation. In the present paper, we employ the derived stochastic representation for $(\hat{V}_{GMV},\hat{\bw}_{GMV}^\top,\hat{R}_{GMV},\hat{s},\hat{\bv}^\top)^\top$ in the derivation of their high-dimensional asymptotic distribution, as well as in obtaining the high-dimensional asymptotic distribution of estimated optimal portfolios.

The theoretical results derived in the paper are based on the assumption that the asset returns are independent and normally distributed. Although this assumption is crucial for the derivation, especially of the finite-sample distributions of the estimated portfolio weights and their estimated characteristics, it is not obviously fulfilled in practical applications, especially, when financial data of daily or higher frequency are considered. On the other side, such assumptions are still appropriate for data taken of weekly or lower frequency, especially when the data is obtained from developed financial markets. For that reason, in the numerical part of the paper, we investigate the robustness of the derived high-dimensional asymptotic results to the violation of both normality and independence
{considering multivariate $t$-distributions and  the CCC-GARCH model} (see, \citet{bollerslev1990modelling}). The results of the simulation study indicate the presence of overestimation for the slope parameter of the efficient frontier and underestimation of the variance of the global minimum variance portfolio under the $t$-model, while the asymptotic normality can still be used for the rest of the desired quantities. Moreover, the presence of autocorrelation between squared asset returns has only minor influence on the derived high-dimensional distributions. It means that, when the data follows a heavy tailed distribution, it is not fully clear what happens with the derived asymptotic distributions, and this case should be considered with care. We expect that the asymptotic results will depend on the fourth moments of the noise variables, and we will study this case deeper in the future papers.

The rest of the paper is organized as follows. In Section 2, we derive the finite-sample joint distribution of  $(\hat{V}_{GMV},\hat{\bw}_{GMV}^\top, \hat{R}_{GMV},\hat{s},\hat{\bv}^\top)^\top$. This result is then used to establish the sampling distributions of the estimated optimal portfolio weights and their estimated characteristics in Section 3. Section 4 presents the asymptotic distributions of the estimated weights derived under the large-dimensional asymptotics. The results of the finite-sample performance of the asymptotic distributions and the robustness analysis to the distributional assumptions imposed on the data-generating process are investigated in Section 5, while final remarks are given in Section 6. The technical derivations are moved to the appendix.

\section{Exact sampling distribution of $\hat{V}_{GMV}$, $\hat{\bw}_{GMV}$, $\hat{R}_{GMV}$, $\hat{s}$, and $\hat{\bv}$}

Throughout the paper we assume that the $p$-dimensional vectors of asset returns $\bx_1,\bx_2,...,\bx_n$ are independent and normally distributed with mean vector $\bmu$ and covariance matrix $\bSigma$, i.e. $\bx_i\sim \mathcal{N}_p(\bmu,\bSigma)$ for $i=1,...,n$. While \citet{fama1976} argued that the distribution of monthly asset returns can be well approximated by the normal distribution, \citet{tu2004data} found no significant impact of heavy tails on the performance of optimal portfolios.

The stochastic representation of $\hat{V}_{GMV}$, $\hat{\btheta}$, $\hat{R}_{GMV}$, $\hat{s}$, and $\hat{\bet}$ is derived in a more general case, namely by considering linear combinations of $\hat{\btheta}$ and $\hat{\bet}$ expressed as
\[\hat{\btheta}=\bL \hat{\bw}_{GMV} \qquad \text{and} \qquad  \hat{\bet}=\bL \hat{\bv},
\]
where $\bL$ is a $k \times p$ matrix of constant with $k<p-1$ and $rank(\bL)=k$.In Sections 2 and 3 we assume that $k<p-1$, while a stronger condition on $k$ is imposed in the derivation of the high-dimensional asymptotic results given in Section 4. Here, we assume that $k$ remains finite, while the portfolio dimension $p$ increases together with the sample size $n$.  The matrix $\bL$ represents the investors interest and views of the portfolio weights. If $\bL$ is selected to be the row vector with one on the $i$th element and zero otherwise, for instance $(1,0,0,...,0)$, then the investor would be interested in the $i$th weight and its distribution. If the distribution of the $i$th weight is centered around zero, then one can start to question whether or not the true weight is actually zero and in turn if the asset should be present in the portfolio.

In the same manner as above, we define the population counterparts of $\hat{\btheta}$ and $\hat{\bet}$ given by
\[\btheta=\bL \bw_{GMV} \qquad \text{and} \qquad \bet=\bL \bv.
\]

Since $\hbmu$ and $\hbSigma$ are independently distributed (cf. \citet{rencher1998multivariate}), the conditional distribution of $(\hat{V}_{GMV}$, $\hat{\btheta}^\top$, $\hat{R}_{GMV}$, $\hat{s}$, $\hat{\bet}^\top)^\top$ under the condition $\hbmu=\tbmu$ is equal to the distribution of $(\hat{V}_{GMV},\hat{\btheta}^\top,\tilde{R}_{GMV},\tilde{s},\tilde{\bet}^\top)^\top$ with
\begin{equation}\label{con_5quant}
\tilde{R}_{GMV}=\frac{\tbmu^\top\hbSigma^{-1}\bi}{\bi^\top \hbSigma^{-1}\bi},~~
\tilde{s}=\tbmu^\top \hbQ \tbmu,~~\text{and}~~
\tilde{\bet}=\frac{\bL\hbQ \tbmu}{\tbmu^\top \hbQ \tbmu},
\end{equation}
while their population counterparts we denote by:
\begin{equation}\label{breve_5quant}
\breve{R}_{GMV}=\frac{\tbmu^\top\bSigma^{-1}\bi}{\bi^\top \bSigma^{-1}\bi},~~
\breve{s}=\tbmu^\top \bQ \tbmu,~~\text{and}~~
\breve{\bet}=\frac{\bL\bQ \tbmu}{\tbmu^\top \bQ \tbmu}.
\end{equation}

Let the symbol $\stackrel{d}{=}$ denote the equality in distribution. In Lemma \ref{th1} we derive the joint distribution of $(\hat{V}_{GMV},\hat{\btheta}^\top,\tilde{R}_{GMV},\tilde{s},\tilde{\bet}^\top)^\top$, whose result is an intermediate step towards Theorem \ref{th2}. For that reason, Lemma \ref{th1} is placed in the Appendix. In Theorem \ref{th2} we present a joint stochastic representation of $\hat{V}_{GMV}$, $\hat{\btheta}$, $\hat{R}_{GMV}$, $\hat{s}$, and $\hat{\bet}$, which will be used in the next section to characterize the distribution of portfolio weights on the efficient frontier. Furthermore, we use the notation $\bt \sim t_p(r)$ to indicate that a random vector $\bt$ of size $p$ follows a standardized multivariate t distribution with $r$ degrees of freedom. Since the multivariate t distribution is not uniquely defined, we state the density to be used. The density and further information   can be found in \citet{gupta2013elliptically}. If $\bt \sim t_p(r)$, then we imply that it has density
\begin{equation}
f_{\bt} (\by) = \frac{\Gamma (p)}{(\pi r)^{p/2} \Gamma(r/2)} \left( 1 + \frac{\by^\top \by}{r}\right)^{- \frac{p+r}{2}},
\end{equation}
where $\Gamma (\cdot)$ is the gamma function. We omit the subindex in $t_p$ when it is a one dimensional t distribution. The proof of Theorem \ref{th2} is given in the appendix.

\begin{theorem}\label{th2}
Let $\bx_1,\bx_2,...,\bx_n$ be independent and normally distributed with mean vector $\bmu$ and covariance matrix $\bSigma$, i.e. $\bx_i\sim \mathcal{N}_p(\bmu,\bSigma)$ for $i=1,...,n$ with $n>p$. Define $\bM=(\bL^\top,\tbmu,\bi)^\top$ and assume that $rank(\bM)=k+2$. Let $\bSigma$ be positive definite. Then, a joint stochastic representation of $\hat{V}_{GMV}$, $\hat{R}_{GMV}$, $\hat{\btheta}$, $\hat{s}$, and $\hat{\bet}$ is given by
\begin{enumerate}[(i)]
\item \label{th2_VG} $\hat{V}_{GMV}\stackrel{d}{=}\frac{V_{GMV}}{n-1}\xi_1$;
\item \label{th2_RG}
$\hat{R}_{GMV}\stackrel{d}{=}R_{GMV}+\sqrt{V_{GMV}}\left(\frac{z_1}{\sqrt{n}}+\sqrt{f}\frac{t_1}{\sqrt{n-p+1}}\right)$;
\item \label{th2_btheta}\begin{eqnarray}
\hat{\btheta}&\stackrel{d}{=}&\btheta+\sqrt{V_{GMV}}
\Bigg(\frac{s\bet+\bz_2/\sqrt{n}}{\sqrt{f}}\frac{t_1}{\sqrt{n-p+1}}\nonumber \\
&+&\left(\bL \bQ \bL^\top-\frac{\left(s\bet+\bz_2/\sqrt{n}\right)\left(s\bet+\bz_2/\sqrt{n}\right)^\top}
{f}\right)^{1/2}
\sqrt{1+\frac{t_1^2}{n-p+1}}\frac{\bt_2}{\sqrt{n-p+2}}\Bigg);\nonumber
\end{eqnarray}
\item \label{th2_s}$\hat{s}\stackrel{d}{=}(n-1)\left(1+\frac{t_1^2}{n-p+1}\right)\frac{f}{\xi_2}$ with
\begin{equation}\label{th2_f}
f=\frac{\xi_3}{n}+ \left(s\bet+ \frac{\bz_2}{\sqrt{n}}\right)^\top\left(\bL \bQ \bL^\top\right)^{-1}\left( s\bet+ \frac{\bz_2}{\sqrt{n}}\right);
\end{equation}
\item \label{th2_bet} \begin{eqnarray}
\hat{\bet}&\stackrel{d}{=}&\frac{s\bet+\bz_2/\sqrt{n}}{f}
+\frac{1}{\sqrt{f}}\left(\bL \bQ \bL^\top-\frac{\left(s\bet+\bz_2/\sqrt{n}\right)\left(s\bet+\bz_2/\sqrt{n}\right)^\top}
{f}\right)^{1/2} \nonumber \\
&\times& \left(\frac{1}{\sqrt{1+\frac{t_1^2}{n-p+1}}}\frac{\bt_2}{\sqrt{n-p+2}}\frac{t_1}{\sqrt{n-p+1}}
+ \left(\bI_k+f\frac{\bt_2\bt_2^\top}{n-p+2}\right)^{1/2}\frac{\bt_3}{\sqrt{n-p+3}} \right)\nonumber
\end{eqnarray}
\end{enumerate}
where $\xi_1\sim \chi^2_{n-p}$, $\xi_2\sim \chi^2_{n-p+2}$, $\xi_3\sim \chi^2_{p-k-1;n\bmu^\top\bA \bmu}$, $z_1\sim \mathcal{N}(0,1)$, $\bz_2\sim \mathcal{N}_k(\mathbf{0},\bL \bQ \bL^\top)$, $t_1 \sim t(n-p+1)$, $\bt_2\sim t_k(n-p+2)$, and $\bt_3 \sim t_k(n-p+3)$ are mutually independent with
\begin{equation}\label{th2_bA}
\bA=\bQ- \bQ \bL^\top \left(\bL \bQ \bL^\top\right)^{-1}\bL \bQ.
\end{equation}
\end{theorem}

The results of Theorem \ref{th2} provides a simple way to simulate observations from the distribution of $\hat{V}_{GMV}$, $\hat{R}_{GMV}$, $\hat{\btheta}$, $\hat{s}$, and $\hat{\bet}$. To simulate observations from the joint distribution, we only need to simulate random variables from well-known distributions. Moreover, the total dimension of independently simulated variables is equal to $(3k+5)$, which is considerably small when direct simulation would imply simulating $p \times p$ matrices from a Wishart distribution and a $p$-dimensional vector from a normal distribution.
To this end, we point out that both the square roots in \eqref{th2_btheta} and \eqref{th2_bet} can be computed analytically, which will further facilitate  speeding up the simulation study. This observation is based on the following two equalities
\begin{equation}
  (\bD - \bb\bb^\top)^{1/2} =  \bD^{1/2} \left(\bI -  c \bD^{-1/2} \bb \bb^\top \bD^{-1/2}  \right)
\end{equation}
where $\bD^{1/2}$ is a symmetric square root of $\bD$, $c=(1-\sqrt{1-\bb^\top \bD^{-1} \bb})/\bb^\top \bD^{-1} \bb$
and
\begin{equation}\label{eqn:mat_sqrt}
  (\bI + \bd\bd^\top)^{1/2} =  \bI+ a \bd \bd^\top
\end{equation}
where $a=(\sqrt{1+\bd^\top \bd}-1)/\bd^\top \bd$. Hence, it holds that
\begin{flalign}\label{sq_root1}
&\left(\bL \bQ \bL^\top-\frac{\left(s\bet+\bz_2/\sqrt{n}\right)\left(s\bet+\bz_2/\sqrt{n}\right)^\top}
{f}\right)^{1/2} &\nonumber \\
=&\left(\bL \bQ \bL^\top\right)^{1/2}\left(\bI_k-
\frac{1-\sqrt{\frac{\xi_3}{nf}}}{f-\frac{\xi_3}{n}}
\left(\bL \bQ \bL^\top\right)^{-1/2}\left(s\bet+\bz_2/\sqrt{n}\right)\left(s\bet+\bz_2/\sqrt{n}\right)^\top\left(\bL \bQ \bL^\top\right)^{-1/2}\right)&
\end{flalign}
and
\begin{equation}\label{sq_root2}
 \left(\bI_k+f\frac{\bt_2\bt_2^\top}{n-p+2}\right)^{1/2}=\bI_k+\left(\sqrt{1+f\frac{\bt_2^\top\bt_2}{n-p+2}}-1\right)\frac{\bt_2\bt_2^\top}{\bt_2^\top\bt_2}.
\end{equation}

In equations \eqref{sq_root1} and \eqref{sq_root2} the matrix inverse and square roots are functions of population quantities. They only need to be computed once, independently of the length of the generated sample. The same argument cannot be performed when simulations of the joint distribution are obtained through simulating the sample covariance matrix and the sample mean vector directly. Hence, Theorem \ref{th2} provides an efficient algorithm to generate samples of arbitrary large size from the joint distribution of $\hat{V}_{GMV}$, $\hat{R}_{GMV}$, $\hat{\btheta}$, $\hat{s}$, and $\hat{\bet}$ relative to simulating $\hat{\bSigma}$ and $\hbmu$ directly. The findings of Theorem \ref{th2} also leads to an efficient way of sampling from the sample distribution of the optimal portfolio weights and their estimated characteristics, which will be discussed in detail in the next section. 

\section{Exact sampling distribution of optimal portfolio weights}
The weights of the optimal portfolios that belong to the efficient frontier have the following structure
\begin{equation}\label{w_g}
\bw_g = \bw_{GMV} + g(R_{GMV},V_{GMV}, s) \mathbf{v}
\end{equation}
with their $k$ linear combinations expressed as
\begin{equation}\label{eqn:w}
\bL\bw_g = \btheta + g(R_{GMV},V_{GMV}, s) \bet,
\end{equation}
where the function $g(R_{GMV},V_{GMV}, s)$ determines a specific type of an optimal portfolio. This function depends on $\bmu$ and $\bSigma$ only through the three quantities $R_{GMV}$, $V_{GMV}$, and $s$, which fully determine the efficient frontier in the mean-variance space. By considering the general form of \eqref{eqn:w}, we are able to cover a number of well-known optimal portfolios: the global minimum variance (GMV) portfolio, the mean-variance (MV) portfolio, the expected maximum exponential utility (EU) portfolio, the tangency (T) portfolio, the optimal portfolio that maximizes the Sharpe ratio (SR), the minimum Value-at-Risk (MVaR) portfolio, and the minimum conditional Value-at-Risk (MCVaR) portfolio, the maximum Value-of-Return (MVoR) portfolio, the maximum conditional Value-of-Return (MCVoR) portfolio, among others. The specific choices of $g(.,.,,)$ for each of these optimal portfolios are provided in Table \ref{tab1}.

\begin{table}[ptbh]
\begin{center}
\begin{tabular}{c|c|c}
\hline\hline
Portfolio & $g(R_{GMV},V_{GMV}, s)$ & Additional quantities\\
\hline \hline
GMV& $0$& \\
\hline
MV& $\mu_0 - R_{GMV}$ & $\mu_0 \in \mathds{R}$ -target expected return\\
\hline
EU& $\gamma^{-1} s$& $ \gamma > 0$ is the risk-aversion coefficient\\
\hline
T& $V_{GMV} s/(R_{GMV}-r_f)$ & $r_f$ is the risk-free return\\
\hline
SR& $V_{GMV} s/ R_{GMV}$&\\
\hline
MVaR& $s\sqrt{V_{GMV} /(z^2_\alpha-s)}$& $z_\alpha=\Phi^{-1}(\alpha)$ \\
\hline
MCVaR& $s\sqrt{V_{GMV} /(k^2_\alpha-s)}$&$k_\alpha=\text{exp}\{-z^2_\alpha/2\}/(2\pi(1-\alpha))$\\
\hline
MVoR& $\frac{(R_{GMV}+v_0)s + \sqrt{z_\alpha^2 s\left((R_{GMV}+v_0)^2+(s-z_\alpha^2) V_{GMV}\right)}}{z_\alpha^2- s}$ & $v_0>0$ is the target value-at-risk\\
\hline
MCVoR& $\frac{(R_{GMV}+k_0)s + \sqrt{k_\alpha^2 s\left((R_{GMV}+k_0)^2+(s-k_\alpha^2) V_{GMV}\right)}}{k_\alpha^2- s}$ & $k_0$ is the target conditional value-at-risk\\
\hline\hline
\end{tabular}
\end{center}
\caption{Choice of the function $g$ for several optimal portfolios. The symbol $\Phi(.)$ denotes the distribution function of the standard normal distribution and $\Phi^{-1}(.)$ stands for its inverse.}%
\label{tab1}%
\end{table}

Let $\hat{\bw}_g$ denote the sample estimator of the optimal portfolio weights given in the general form as in \eqref{eqn:w}, which is obtained by plugging the sample mean vector and the sample covariance matrix instead of the unknown population counterparts. The $k$ linear combinations of the optimal portfolio weights are estimated by
\begin{equation}\label{eqn:what}
\bL\hat{\bw}_g= \hbtheta +  g(\hat{R}_{GMV}, \hat{V}_{GMV}, \hat{s}) \hbet.
\end{equation}

By Theorem \ref{th2} the exact sampling distribution of \eqref{eqn:what} is derived in terms of its stochastic representation. The results are summarized in Theorem \ref{th3}, whose proof follows from Theorem \ref{th2}.

\begin{theorem}\label{th3}
  Under the conditions of Theorem \ref{th2}, it holds that
\begin{flalign}
\bL \hat{\bw}_g &\stackrel{d}{=} \btheta + \left( \sqrt{\frac{V_{GMV}}{f}} \frac{t_1}{\sqrt{n-p+1}} +
\frac{g(\hat{R}_{GMV}, \hat{V}_{GMV}, \hat{s})}{f} \right) \left(s\bet+\bz_2/\sqrt{n} \right) \nonumber\\
&+ \left( \bL \bQ \bL^\top-\frac{\left(s\bet+\bz_2/\sqrt{n}\right)\left(s\bet+\bz_2/\sqrt{n}\right)^\top}{f} \right)^{1/2}\nonumber \\
 &\times\Bigg(\sqrt{V_{GMV}} \sqrt{1+\frac{t_1^2}{n-p+1} } + \frac{g(\hat{R}_{GMV}, \hat{V}_{GMV}, \hat{s})}{\sqrt{f}} \frac{t_1/\sqrt{n-p+1}}{\sqrt{1+\frac{t_1^2}{n-p+1}}}
 \Bigg)\frac{\bt_2}{\sqrt{n-p+2}}\nonumber\\
  &+ \frac{g(\hat{R}_{GMV}, \hat{V}_{GMV}, \hat{s})}{\sqrt{f}}  \left( \bI_k + f\frac{\bt_2 \bt_2^\top}{n-p+2}\right)^{1/2} \frac{\bt_3}{\sqrt{n-p+3}}\label{eqn:th3_samp}
\end{flalign}
where the joint stochastic representation of $\hat{V}_{GMV}$, $\hat{R}_{GMV}$ and $\hat{s}$ is given in \eqref{th2_VG}-\eqref{th2_bet} of Theorem \ref{th2}.
\end{theorem}

From Theorem \ref{th3} we can derive a number of important results. First, it provides a complete characterization of the sampling distribution of the estimators for the optimal portfolio weights. This distribution can be assessed by drawing samples with independent observations from the derived stochastic representation of a relatively large size and then applying the well-established statistical methods for estimating the distribution function, the density, the moments, etc. Second, the obtained stochastic representation in Theorem \ref{th3} provides an efficient way for generating samples from the finite-sample distribution of $\bL \hat{\bw}_g$ following the discussion provided in Section 2 after Theorem \ref{th2}, which is based on drawing independent realizations from well-known univariate and multivariate distributions. To this end, we note that the two square roots in (\ref{eqn:th3_samp}) should be computed as given by \eqref{sq_root1} and \eqref{sq_root1}. Similarly to the prior discussion, using these simplifications the derived stochastic representation can be rewritten to include matrix inverses and square roots of population quantities. Once more, these objects should only be computed once during the whole simulation study. Third, for the chosen values of the population quantities used in the simulation study, we can construct concentration sets of optimal portfolio weights. Fourth, an important probabilistic result about the sampling distribution of $\bL \hat{\bw}_g$ follows directly from the derived stochastic representation, namely that the finite-sample distribution of $\bL \hat{\bw}_g$ depends on the population mean vector $\bmu$ and the population covariance matrix $\bSigma$ through $R_{GMV}$, $V_{GMV}$, $s$, $\btheta$, $\bet$, and $\bL \bQ \bL$. To sample from the distribution of $\bL \hat{\bw}_g$ we only need to fix these seven quantities. In particular, in the case of a single linear combination, i.e. when $k=1$, we only have to fix six univariate quantities independently of the dimension $p$ of the data-generating process.

In a similar way, we derive statistical inference for the estimated characteristics of optimal portfolio with weights $\hat{\bw}_g$ as given by \eqref{w_g}. The expected return of the optimal portfolio with the weights \eqref{w_g} is given by
\begin{equation}\label{eqn:R}
R_g = R_{GMV} + g(R_{GMV},V_{GMV}, s),
\end{equation}
while its variance is
\begin{equation}\label{eqn:V}
V_g = V_{GMV} + \frac{g(R_{GMV},V_{GMV}, s)^2}{s}.
\end{equation}
Similarly, the Value-at-Risk (VaR), the Conditional Value-at-Risk (CVaR), the Value-of-Return (VoR) and the Conditional Value-of-Return (CVoR) are computed by
\begin{eqnarray}
VaR_g &=& -\left( R_{GMV} + g(R_{GMV},V_{GMV}, s) \right) - z_\alpha \sqrt{V_{GMV} + \frac{g(R_{GMV},V_{GMV}, s)^2}{s}}  ,\label{eqn:VaR}\\
CVaR_g &=&  -\left( R_{GMV} + g(R_{GMV},V_{GMV}, s) \right) - k_\alpha \sqrt{V_{GMV} + \frac{g(R_{GMV},V_{GMV}, s)^2}{s}},\label{eqn:CVaR}
\end{eqnarray}
and by symmetry
\begin{eqnarray}
VoR_g &=& \left( R_{GMV} + g(R_{GMV},V_{GMV}, s) \right) - z_\alpha \sqrt{V_{GMV} + \frac{g(R_{GMV},V_{GMV}, s)^2}{s}} ,\label{eqn:VoR}\\
CVoR_g &=& \left( R_{GMV} + g(R_{GMV},V_{GMV}, s) \right) - k_\alpha \sqrt{V_{GMV} + \frac{g(R_{GMV},V_{GMV}, s)^2}{s}} .\label{eqn:CVoR}
\end{eqnarray}
Equation \eqref{eqn:VaR} and \eqref{eqn:CVaR} describe the VaR and CVaR respectively for a specific portfolio and not a single asset. Since risk measures simply measure the risk of random objects, the interpretation of the VaR remains the same as in the case of the univariate case. The VaR measures the loss at a certain confidence level and specifies a quantile. The CVaR specifies the mean loss when a loss larger {than} (or equal to) the VaR occurs, a tail-conditional expectation, hence its name Conditional Value-at-Risk. By a change of sign, the same holds for the CVoR and VoR.

Inserting the sample mean vector and the sample covariance matrix in \eqref{eqn:R}-\eqref{eqn:CVoR} instead of the population counterparts, we get the sample estimators of the optimal portfolio characteristics. The application of Theorem \ref{th2} leads to the statement of their (joint) sampling distribution, which is presented in Theorem \ref{th4}

\begin{theorem}\label{th4}
Under the conditions of Theorem \ref{th2}, the stochastic representation of the estimated characteristic of optimal portfolio are obtained as in \eqref{eqn:R}-\eqref{eqn:CVoR} where $R_{GMV}$, $V_{GMV}$, and $s$ are replaced by their sample counterparts $\hat{R}_{GMV}$, $\hat{V}_{GMV}$, and $\hat{s}$
with
\begin{eqnarray*}
\hat{V}_{GMV}&\stackrel{d}{=}& \frac{V_{GMV}}{n-1}\xi,\\
\hat{R}_{GMV}&\stackrel{d}{=}& R_{GMV} + \sqrt{\frac{V_{GMV}}{n} \left(1+\frac{p-1}{n-p+1} \psi \right)}z, \\
\hat{s}&\stackrel{d}{=}&\frac{(n-1)(p-1)}{n(n-p+1)} \eta,
\end{eqnarray*}
where $\xi\sim \chi^2_{n-p}$, $\psi \sim F(p-1, n-p+1, ns)$, $z \sim N(0,1)$ are mutually independent.
\end{theorem}

The proof of Theorem \ref{th4} is given in the appendix. It has to be noted that the joint distribution of all six estimators $(\hat{R}_g,\hat{V}_g,\widehat{VaR}_g,\widehat{CVaR}_g,\widehat{VoR}_g,\widehat{CVoR}_g)$ is completely determined by three mutually independent random variables $\xi$, $\psi$, and $z$ with the standard marginal univariate distribution. Moreover, it depends on the unknown population mean vector and covariance matrix over three univariate quantities $R_{GMV}$, $V_{GMV}$, and $s$. These uniquely determine the whole efficient frontier in the mean-variance space. To this end, the stochastic representation derived for the estimated optimal portfolio characteristics appears to be simpler than the one obtained in Theorem \ref{th3} for the corresponding estimator of the optimal portfolio weights. Similarly, the independent realizations from the joint distribution of $(\hat{R}_g,\hat{V}_g,\widehat{VaR}_g,\widehat{CVaR}_g,\widehat{VoR}_g,\widehat{CVoR}_g)$ can be drawn efficiently by employing the results of Theorem \ref{th4}.

Another interesting financial application of the derived theoretical findings of Theorem \ref{th4} is present in the case of the EU portfolio, whose sample expected return and sample variance possess the following stochastic representations:
\begin{eqnarray}
 \hat{R}_{EU} & \stackrel{d}{=}& \hat{R}_{GMV} + \gamma^{-1} \hat{s}, \label{eqn:eu_rg} \\
 \hat{V}_{EU} & \stackrel{d}{=}& \hat{V}_{GMV} + \gamma^{-2}\hat{s}. \label{eqn:eu_vg}
\end{eqnarray}
It appears that $\hat{R}_{EU}$ and $\hat{R}_{EU}$ are conditionally independent, given the estimated slope parameter of the efficient frontier $\hat{s}$. In the limit case, when the risk aversion coefficient $\gamma$ tends to infinity and the EU portfolio tends towards the vertex of the efficient frontier, the two estimated portfolio characteristics are unconditionally independent. In all other cases, the dependence between them is fully captured by the estimated geometry of the efficient frontier.

\section{High-dimensional asymptotic distributions}
The derived stochastic representations of Sections 3 and 4 are also very useful in the derivation of the asymptotic distributions of the estimators of optimal portfolio weights and their estimated characteristics. To this end, we note that the same approach can be used independently whether the dimension of the data generating process $p$ is assumed to be fixed or allowed to grow together with the sample size $n$. These two regimes have been intensively discussed in statistical literature. The former asymptotic regime, i.e. with fixed $p$, is called the ``standard asymptotics'' (see, e.g., \citet{le2012asymptotics}). Here, both the sample mean and the sample covariance matrix are proven to be consistent estimators for the corresponding population counterparts. Challenges arise when $p$ is comparable to $n$, i.e. both the dimension $p$ and the sample size $n$ tend to infinity while their ratio $p/n$ tends to a positive constant $c \in [0,1)$, the so-called concentration ratio. It is called ``large dimensional asymptotics'' or ``Kolmogorov asymptotics'' (c.f., \citet{buhlmann2011statistics}, \citet{cai2011analysis}), while the case $c=0$ corresponds to the standard asymptotics.

Although,  a large amount of research has been done on the asymptotic behavior of functionals which only include the sample mean vector or the sample covariance matrix under the high-dimensional asymptotics (see, e.g., \citet{bai2010spectral}, \citet{cai2010}, \citet{wang2014}, \citet{BodnarGuptaParolya2016}, \citet{bodnar2019optimalmean}, \citet{bodnardetteparolya2019}), the situation becomes more complicated when both the sample mean vector and the sample (inverse) covariance matrix are present in the expressions. The problem is still unsolved and attracts both  researchers and  practitioners. In this section, we show how the derived stochastic representations of Sections 2 and 3 can be employed in the derivation of the high-dimensional asymptotic distributions of the estimated optimal portfolios and their characteristics. The main advantage of the suggested approach based on the stochastic representations is {the clear  separation between} the deterministic quantities and the stochastic ones. By using the stochastic representation we can determine the joint asymptotic distributions of the latter.

Throughout this section, we will impose the following technical conditions on the functions involving the population mean vector and the population covariance matrix:
\begin{description}
\item[(A1)] There exist $m$ and $M$, such that
\begin{equation}\label{A1_1}
0<m\le \bmu^\top \bSigma^{-1} \bmu \le M <\infty
\quad\text{and} \quad
0<m\le \bi^\top \bSigma^{-1} \bi \le M <\infty
\end{equation}
uniformly in $p$. Moreover, for a linear combination of optimal portfolio weights determined by the $p$-dimensional vector $\mathbf{l}$ it holds uniformly in $p$ that
\begin{equation}\label{A1_2}
0<m\le \mathbf{l}^\top \bSigma^{-1} \mathbf{l} \le M <\infty.
\end{equation}
\end{description}

Assumption (A1) ensures that the efficient frontier $R_{GMV}$, $V_{GMV}$, and $s$ as well as the components of $k$ linear combinations of optimal portfolio weights $\mathbf{L} \mathbf{w}_g$ are all finite numbers in higher dimensions. The financial interpretation shows  that even though we have an infinite amount of assets, we should not be able to gain an infinite amount of return for any amount of risk taken, i.e., the slope of the efficient frontier is bounded. It also states that the variance of the GMV portfolio is finite and bounded from zero, which can also be interpreted as  investing in a market (regardless of how big it is), should imply some risk, but it can be neither infinite nor zero.
Mathematically, it may happen depending on $\bmu$ and $\bSigma$ that some quantities of $R_{GMV}$, $V_{GMV}$, $s$, and $\mathbf{L} \mathbf{w}_g$ tend to infinity as $p$ increases. In such cases, one should replace the constants $m$ and $M$ in \eqref{A1_1} and \eqref{A1_2} by $p^{\kappa}m$ and $p^{\kappa}M$ for some $\kappa>0$. This approach would lead only to minor changes in the expressions of the derived asymptotic covariance matrices in this section, where some terms might disappear (see, e.g., \citet{bodnar2019central} for a similar discussion). 

To this end, by an abuse of notations, we use the same notations for the functions involving the population mean vector $\bmu$ and the population covariance matrix $\bSigma$ and their corresponding deterministic limits. For instance, $\bmu^\top \bSigma^{-1} \bmu$ will also be used to denote the limit $\lim_{p \to \infty} \bmu^\top \bSigma^{-1} \bmu$. {The interpretation of the quantities becomes clear from the text where they are used.}

\subsection{High-dimensional asymptotic distribution of $\hat{V}_{GMV}$, $\hat{R}_{GMV}$, $\hat{\btheta}$, $\hat{s}$, and $\hat{\bet}$}
Before presenting the high-dimensional asymptotic results for the estimated optimal portfolio weights and their characteristic, we derive the asymptotic stochastic representation for the five quantities $\hat{V}_{GMV}$, $\hat{R}_{GMV}$, $\hat{\btheta}$, $\hat{s}$, and $\hat{\bet}$. It is presented in Theorem~\ref{th5} in terms of several independently normally distributed random variables/vectors. Such a presentation allows also to characterize the asymptotic dependence structure $\hat{V}_{GMV}$, $\hat{R}_{GMV}$, $\hat{\btheta}$, $\hat{s}$, and $\hat{\bet}$ as well as to derive the expression of the asymptotic covariance matrix which is given after Theorem \ref{th5}.

\begin{theorem}\label{th5}
 Under the conditions of Theorem~\ref{th2} and Assumption (\textbf{A1}), it holds that {\small
 \begin{enumerate}[(i)]
  \item $\sqrt{n-p}\left(\hat{V}_{GMV} - \frac{1-p/n}{1-1/n} V_{GMV} \right) \stackrel{d}{\rightarrow} \sqrt{2}(1-c) V_{GMV} u_1$,
  \item $\sqrt{n-p}\left( \hat{R}_{GMV}-R_{GMV}\right) \stackrel{d}{\rightarrow} \sqrt{V_{GMV}}\left(\sqrt{1-c} u_4 + \sqrt{s+c} u_5 \right)$,
  \item $\sqrt{n-p}\left( \hat{\btheta}-\btheta\right) \stackrel{d}{\rightarrow} \sqrt{V_{GMV}}\left(\frac{s u_5 }{\sqrt{s+c}}\bet+ \left(\bL \bQ \bL^\top - \frac{s^2}{s+c}\bet \bet^\top\right)^{1/2}\bu_6 \right)$,
  \item $ \sqrt{n-p}\left( \hat{s} - \frac{(s+p/n)(1-1/n)}{1-p/n+2/n} \right)\stackrel{d}{\rightarrow}  \frac{1}{1-c}\Bigg( \sqrt{2(1-c)\left(c+2\bmu^\top \bA \bmu \right)}u_2 + 2s\sqrt{(1-c)} \bet^\top (\bL\bQ\bL^\top)^{-1/2} \bu_3 + \sqrt{2}(s+c)u_7\Bigg)$,
  \item $\sqrt{n-p}\left(\hat{\bet} - \frac{s}{s+p/n}\bet \right)\stackrel{d}{\rightarrow}  \frac{1}{\sqrt{s+c}}\left(\bL \bQ \bL^\top - \frac{s^2}{s+c}\bet \bet^\top\right)^{1/2}\bu_8+ \frac{\sqrt{1-c}}{(s+c)} \Bigg(\bL \bQ \bL^\top - 2\frac{s^2}{s+c}\bet \bet^\top \Bigg) (\bL \bQ \bL^\top)^{-1/2} \bu_3$\\ $-\frac{s\sqrt{2(1-c)\left(c+2\bmu^\top \bA \bmu \right)}u_2}{(s+c)^2}\bet$
  \end{enumerate}
  }
for $p/n \rightarrow c\in [0,1)$ as $n \rightarrow \infty$ where $u_1, u_2, \bu_3, u_4, u_5, \bu_6, u_7, \bu_8$  are mutually independent, $u_1, u_2, u_4, u_5, u_7 \sim N(0,1)$ and $ \bu_3, \bu_6, \bu_8 \sim N_k(\mathbf{0}, \bI_k)$.
\end{theorem}

Several interesting results are summarized in the statement of Theorem \ref{th5}, whose proof is given in the appendix. We observe that three quantities related to the estimators of the weights and of the characteristics of the GMV portfolio, the vertex point on the efficient frontier, are asymptotically independent of the estimated slope parameter of the efficient frontier $\hat{s}$ and the self-financing portfolio $\hat{\bet}$. However, these two are not asymptotically independent. This may not be so surprising since it is implicitly present in the estimated weights of the self-financing portfolio $\bv$ from its definition \eqref{sample_5quant}. Moreover, the sample variance of the GMV portfolio appears to be asymptotically independent of its estimated expected return $\hat{R}_{GMV}$ and the estimator of the weights $\hat{\btheta}$ following the finite-sample findings of Theorem \ref{th2}. However, it is surprising that the covariance between $\hbtheta$ and $\hat{R}_{GMV}$ is partly determined by the self-financing portfolio $\bet$ due to the deterministic expression close to $u_5$ in the asymptotic stochastic representations of $\sqrt{n-p}\left( \hat{R}_{GMV}-R_{GMV}\right)$ and $\sqrt{n-p}\left( \hat{\btheta}-\btheta\right)$. Finally, the direct application of the derived stochastic representations in Theorem \ref{th5} leads to the expression of the asymptotic covariance matrix as given in Corollary \ref{cor1}.

\begin{corollary}\label{cor1}
 Under the conditions of Theorem~\ref{th2} and Assumption (\textbf{A1}), it holds that
\begin{equation*}
 \sqrt{n-p}\begin{pmatrix}
 \hat{V}_{GMV} - \frac{1-p/n}{1-1/n} V_{GMV} \\
 \hat{R}_{GMV}- R_{GMV} \\
 \hat{\btheta} - \btheta \\
 \hat{s}-\frac{(s+p/n)(1-1/n)}{1-p/n+2/n} \\
 \hat{\bet} - \frac{s}{s+p/n}\bet
 \end{pmatrix} \rightarrow N_{2k+3}\left(\mathbf{0},\mathbf{\Xi} \right)
\end{equation*}
with
\begin{equation*}
\mathbf{\Xi}=
 \begin{pmatrix}
 2 V^2_{GMV} (1-c)^2 & 0 & 0 & 0 & 0 \\
 0 & V_{GMV}(1+s) &  V_{GMV} s \bet^\top & 0 & 0  \\
 0 & V_{GMV} s \bet &   V_{GMV} \bL \bQ \bL^\top & 0 & 0  \\
 0 & 0 & 0 & \Xi_{s,s} & \boldsymbol{\Xi}^\top_{s, \bet}  \\
 0 & 0 & 0 & \boldsymbol{\Xi}_{s, \bet}  &  \boldsymbol{\Xi}_{\bet,\bet}
 \end{pmatrix}
\end{equation*}
for $p/n \rightarrow c\in [0,1)$ as $n \rightarrow \infty$ where
 \begin{eqnarray}
\Xi_{s,s}&=& \frac{2(c+2s)}{(1-c)} + 2\frac{(s+c)^2}{(1-c)^2}, \label{eqn:s_asymp}\\
\boldsymbol{\Xi}_{\bet,\bet}&=& \frac{s+1}{(s+c)^2}\bL\bQ\bL^\top - \frac{s^2(2c(1-c)+(s+c)^2)}{(s+c)^4}\bet \bet^\top, \label{eqn:bet_asymp} \\
\boldsymbol{\Xi}_{s, \bet} &=& \frac{2s\left(2c- s + 4\bmu^\top \bA \bmu) \right)}{(s+c)^2}\bet. \nonumber
\end{eqnarray}
\end{corollary}

\subsection{High-dimensional asymptotic distribution of optimal portfolio weights}

The results of Theorem \ref{th5} are used to derive the high-dimensional asymptotic distribution of linear combinations of the estimated optimal portfolio weights $\hat{\bw}_g$ as well as of the corresponding estimated characteristics of the portfolios given in Section 3. Throughout this section we assume that the number of linear combinations $k$ is finite.

Let
\begin{equation}\label{lam}
\hat{\boldsymbol{\lambda}}=(\hat{R}_{GMV}, \hat{V}_{GMV}, \hat{s})^\top
\quad \text{and} \quad
\boldsymbol{\lambda}=\left(R_{GMV}, (1-c)V_{GMV}, \frac{s+c}{1-c}\right)^\top
\end{equation}
where the results of Theorem \ref{th5} show that
\begin{eqnarray*}
\hat{R}_{GMV} - R_{GMV}&=& o_P(1),\\
\hat{V}_{GMV} - (1-c)V_{GMV}&=& o_P(1),\\
\hat{s}-\frac{s+c}{1-c}&=& o_P(1),
\end{eqnarray*}
where $o_P(1){\to}0$ for $p/n \rightarrow c\in[0,1)$ as $n \rightarrow \infty$.

Throughout this section it is assumed that the function $g(x,y,z)$ is differentiable with first order continuous derivatives and define
\begin{eqnarray*}
g_1(x_0,y_0,z_0)&=&\frac{\partial g(x,y,z)}{\partial x}\Bigg|_{(x,y,z)=(x_0,y_0,z_0)}, \\
g_2(x_0,y_0,z_0)&=&\frac{\partial g(x,y,z)}{\partial y}\Bigg|_{(x,y,z)=(x_0,y_0,z_0)}, \\
g_3(x_0,y_0,z_0)&=&\frac{\partial g(x,y,z)}{\partial z}\Bigg|_{(x,y,z)=(x_0,y_0,z_0)}.
\end{eqnarray*}
The asymptotic distribution of $\bL\hat{\bw}_g$ is given in Theorem \ref{th6}, with the proof presented in the appendix.

\begin{theorem}\label{th6}
Let $g(.,.,.)$ be differentiable with first order continuous derivatives. Then, under the conditions of Theorem~\ref{th2} and Assumption (\textbf{A1}), we get
\begin{eqnarray}
&& \sqrt{n-p}\left(\bL\hat{\bw}_g - \left(\btheta + \frac{s g\left(\blambda\right)}{s+p/n}\bet \right) \right)
 \stackrel{d}{\rightarrow} N_k(\mathbf{0}, \bOmega_{\bL,g})
\end{eqnarray}
for $p/n \rightarrow c\in [0,1)$ as $n \rightarrow \infty$ with
\begin{eqnarray}
&&\bOmega_{\bL,g} =  \bigg( \left(\frac{1-c}{s+c} + g\left(\blambda\right)\right)
\frac{ g\left(\blambda\right)}{s+c} + V_{GMV}  \bigg) \bL \bQ \bL^\top
+ s^2\Bigg\{
	2\frac{(1-c)^2 V^2_{GMV}}{(s+c)^2}  g_2\left(\blambda\right)\nonumber \\
&	+&  \left(\frac{g_3\left(\blambda\right)}{1-c} - \frac{ g\left(\blambda\right)}{s+c} \right)^2 \frac{2(1-c)c}{(s+c)^2} + \frac{4(1-c)}{(s+c)^2}
 \Bigg[  g\left(\blambda\right) \left(\frac{g_3\left(\blambda\right)}{1-c}  - \frac{g\left(\blambda\right)}{s+c} \right)
 + s\left(\frac{g_3\left(\blambda\right) }{1-c}  - \frac{ g\left(\blambda\right)}{s+c} \right)^2  \Bigg]\nonumber \\
&	+ & \frac{ V_{GMV}(1-c) }{(s+c)^2} g_1\left(\blambda\right)^2 +  \frac{V_{GMV}}{(s+c)} g_1\left(\blambda\right) + \frac{2}{1-c} g_3\left(\blambda\right)^2 -\frac{ g\left(\blambda\right)^2}{(s+c)^2} \Bigg\} \bet \bet^\top. \label{Omega_g}
\end{eqnarray}
\end{theorem}

The results of Theorem~\ref{th6} are derived for the finite number $k$ of linear combinations of the estimated optimal portfolio weights $\hat{\bw}_g$ and, consequently, they cannot be used to verify the consistency of the whole vector $\hat{\bw}_g$. Several consistency results about estimated optimal portfolio weights can be found in \citet{ao2019approaching}.

In the special case of the EU portfolio we get $g(x,y,z)=\gamma^{-1} z$, $g_1(x,y,z)=g_2(x,y,z)=0$, and
$$
\frac{g_3\left(\blambda\right)}{1-c} - \frac{ g\left(\blambda\right)}{s+c}  = \frac{1}{1-c} \gamma^{-1} - \frac{\gamma^{-1} (s+c)}{(1-c)(s+c)}=0.
$$
As a result, the asymptotic covariance matrix of $\bL\hbw_{EU}$ is expressed as
 \begin{flalign}
\bOmega_{\bL,EU}&=\bigg( \left(\frac{1-c}{s+c} + \gamma^{-1} \frac{s+c}{1-c} \right)\frac{\gamma^{-1}}{1-c} + V_{GMV}  \bigg) \bL \bQ \bL^\top
+\frac{(1-2c)\gamma^{-2}s^2}{(1-c)^2}\bet \bet^\top.
\end{flalign}

In the same way, the high-dimensional asymptotic distribution of the estimated optimal portfolio characteristics is obtained. Following (\ref{eqn:R})-(\ref{eqn:CVoR}), $(R_g,V_g,VaR_g,CVaR_g,VoR_g,CVoR_g)$ are functions of $R_{GMV}$, $V_{GMV}$, and $s$ only. On the other hand, Theorem \ref{th5} determines the joint high-dimensional asymptotic distribution of $\hat{R}_{GMV}$, $\hat{V}_{GMV}$, and $\hat{s}$ expressed as
\begin{equation*}
 \sqrt{n-p}\begin{pmatrix}
 \hat{R}_{GMV}- R_{GMV} \\
 \hat{V}_{GMV} - \frac{1-p/n}{1-1/n} V_{GMV} \\
 \hat{s}-\frac{(s+p/n)(1-1/n)}{1-p/n+2/n} \\
 \end{pmatrix} \rightarrow N_3\left(\mathbf{0},\mathbf{\Xi}_{RVs} \right)
\end{equation*}
for $p/n \rightarrow c\in [0,1)$ as $n \rightarrow \infty$ with
\begin{equation*}
\mathbf{\Xi}_{RVs}=
 \begin{pmatrix}
 V_{GMV}(1+s) & 0 & 0 \\
 0 & 2 V^2_{GMV} (1-c)^2 &  0  \\
 0 & 0 & \frac{2(c+2s)}{(1-c)} + 2\frac{(s+c)^2}{(1-c)^2} \\
 \end{pmatrix},
\end{equation*}
which shows that $(\hat{R}_{GMV},\hat{V}_{GMV},\hat{s})$ are asymptotically independently distributed.

The characteristics in equations \eqref{eqn:R} through \eqref{eqn:CVoR} can be viewed as transformations of $R_{GMV},V_{GMV},s$. In a similar fashion to Theorem \ref{th6}, we can also construct high-dimensional asymptotic distributions for these. Let $h_{g,i}(R_{GMV},V_{GMV},s)$ denote the $i$-th characteristic of the optimal portfolio with the weights $\bw_g$ and let
 $h_{g,i}\left(\hat{\blambda}\right)$  denote   its corresponding estimate, where $\lambda$ is defined in \eqref{lam}. In the example of characteristics given by equations \eqref{eqn:R} through \eqref{eqn:CVoR}, $i$ would range from $1$ to $6$. Let the $j$-th first order partial derivative of $h_{g,i}(.)$ at $\blambda$ be denoted by $h_{g,i;j}\left(\blambda\right)$. We get the following result of the high-dimensional distribution of estimated optimal portfolio characteristic, whose proof is obtained from the proof of Theorem \ref{th6}.

\begin{theorem}\label{th7}
Let $h_{g,i}(.,.,.)$, $i=1,...,q$, be differentiable with first order continuous derivatives. Then, under the conditions of Theorem~\ref{th2} and Assumption (\textbf{A1}), we get for $p/n \rightarrow c\in[0,1)$ as $n \rightarrow \infty$
\begin{equation*}
 \sqrt{n-p}\begin{pmatrix}
 h_{g,1}\left(\hat{\blambda}\right)- h_{g,1}\left(\blambda\right) \\
 \vdots \\
 h_{g,q}\left(\hat{\blambda}\right)- h_{g,q}\left(\blambda\right) \\
 \end{pmatrix} \rightarrow N_q\left(\mathbf{0},\mathbf{\Xi}_{h} \right)
\end{equation*}
 with $\mathbf{\Xi}_{h}=\left(\Xi_{h;ij}\right)_{i,j=1,...,q}$ where
\begin{equation}\label{Xi_h}
\Xi_{h;ij}=\sum_{l=1}^{3} \Xi_{RVs;ll} h_{g,i;l}\left(\blambda\right)h_{g,j;l}\left(\blambda\right).
\end{equation}
\end{theorem}
\subsection{Interval estimation and high-dimensional test theory}

The results of Theorems \ref{th6} and \ref{th7} indicate that both $\bL\hat{\bw}_g$ and $h_{g,i}\left(\hat{\blambda}\right)$, $i=1,...,q$, are not consistent estimators for $\bL\hat{\bw}_g$ and $h_{g,i}\left(R_{GMV},V_{GMV},s\right)$, $i=1,...,q$, respectively. While the asymptotic bias of the sample estimator of linear combinations of optimal portfolio weights is $\left(\frac{s}{s+c}g\left(\blambda\right)-g\left(R_{GMV},V_{GMV},s\right)\right)\bet$, the asymptotic bias in the estimator of the $i$-th portfolio characteristic is $h_{g,i}\left(\blambda\right)-h_{g,i}\left(R_{GMV},V_{GMV},s\right)$.

On the other hand, the results of Theorem \ref{th5} already provide consistent estimators for $V_{GMV}$, $R_{GMV}$, $\btheta$, $s$, and $\bet$. Namely, they are given by
\begin{eqnarray}
\hat{V}_{GMV;c} &=& \frac{\hat{V}_{GMV}}{1-p/n} \label{hV_c},\\
\hat{R}_{GMV;c}&=& \hat{R}_{GMV} \label{hR_c},\\
\hat{\btheta}_c &=& \hat{\btheta} \label{hbtheta_c},\\
\hat{s}_c&=& \frac{n-p}{n}\left(\hat{s}-\frac{p}{p+n}\right) \label{hs_c},\\
\hat{\bet}_c&=& \frac{\hat{s}_c+p/n}{\hat{s}_c}\hat{\bet}\label{hbet_c}.
\end{eqnarray}
Combining these equalities, we derive consistent estimators for $\bL\hat{\bw}_g$ and $h_{g,i}\left(R_{GMV},V_{GMV},s\right)$ expressed as
\begin{eqnarray}
\bL\hat{\bw}_{g;c}&=&\hat{\btheta} +g\left(\hat{R}_{GMV;c},\hat{V}_{GMV;c},\hat{s}_c\right) \hat{\bet}_c~\text{and}\\
\hat{h}_{g,i,c}&=&h_{g,i}\left(\hat{R}_{GMV;c},\hat{V}_{GMV;c},\hat{s}_c\right).
\end{eqnarray}

In Theorem \ref{th8}, the asymptotic covariance matrices of the consistent estimators of optimal portfolio weights and their characteristics are present.
\begin{theorem}\label{th8}
Let $\blambda=(R_{GMV},V_{GMV},s)^\top$. Then, under the conditions of Theorems \ref{th6} and \ref{th7}, it holds that
\begin{enumerate}[(a)]
\item $ \sqrt{n-p}\left(\bL\hat{\bw}_{g;c} - \bL\bw_g \right) \stackrel{d}{\rightarrow} N_k(\mathbf{0}, \bOmega_{\bL,g,c})
$
for $p/n \rightarrow c\in[0,1)$ as $n \rightarrow \infty$ with
{\small
\begin{eqnarray}\label{Omega_gc}
\bOmega_{\bL,g,c}&=& \bigg( \left(\frac{1-c}{s+c} + \frac{s+c}{s}g\left(\blambda_0\right)\right)
   \frac{ g\left(\blambda_0\right)}{s} + V_{GMV}  \bigg) \bL \bQ \bL^\top\nonumber\\
  &+&s^2\Bigg\{
	2\frac{(1-c) V^2_{GMV}}{s(s+c)}  g_2\left(\blambda_0\right)\nonumber
+  \left(\frac{g_3\left(\blambda_0\right)(s+c)}{s} - \frac{ g\left(\blambda_0\right)}{s} \right)^2 \frac{2(1-c)c}{(s+c)^2}\nonumber\\
& +& \frac{4(1-c)}{(s+c)^2}
 \Bigg[ \frac{s+c}{s} g\left(\blambda_0\right) \left(\frac{g_3\left(\blambda_0\right)(s+c)}{s}  - \frac{g\left(\blambda_0\right)}{s} \right)
 + s\left(\frac{g_3\left(\blambda_0\right)(s+c)}{s}  - \frac{ g\left(\blambda_0\right)}{s} \right)^2  \Bigg]\nonumber \\
&	+ & \frac{ V_{GMV}(1-c) }{s^2} g_1\left(\blambda_0\right)^2 +  \frac{V_{GMV}}{s} g_1\left(\blambda_0\right) + \frac{2(1-c)(s+c)^2}{s^2} g_3\left(\blambda_0\right)^2 -\frac{ g\left(\blambda_0\right)^2}{s^2} \Bigg\} \bet \bet^\top
 \nonumber;
\end{eqnarray}
}
\item for $p/n \rightarrow c\in[0,1)$ as $n \rightarrow \infty$ holds
\begin{equation*}
 \sqrt{n-p}\begin{pmatrix}
 \hat{h}_{g,1,c}- h_{g,1}\left(\blambda_0\right) \\
 \vdots \\
 \hat{h}_{g,q,c}- h_{g,q}\left(\blambda_0\right) \\
 \end{pmatrix} \rightarrow N_q\left(\mathbf{0},\mathbf{\Xi}_{h,c} \right)~~\text{with~~ $\mathbf{\Xi}_{h,c}=\left(\Xi_{h,c;ij}\right)_{i,j=1,...,q}$ }
\end{equation*}
where
\begin{equation}
\Xi_{h,c;ij}=V_{GMV}(1+s) h_{g,i;1}\left(\blambda_0\right)h_{g,j;1}\left(\blambda_0\right)
+2 V^2_{GMV} h_{g,i;2}\left(\blambda_0\right)h_{g,j;2}\left(\blambda_0\right)
+\left(2s^2+4s+2c\right)h_{g,i;3}\left(\blambda_0\right)h_{g,j;3}\left(\blambda_0\right).\label{Xi_hc}
\end{equation}
\end{enumerate}
\end{theorem}

Since both $\bOmega_{\bL,g,c}$ and $\mathbf{\Xi}_{h,c}$ depend on unobservable quantities, we have to estimate them consistently under the high-dimensional asymptotic regime when confidence regions for the optimal portfolio weights and for the optimal portfolio characteristics are derived.

Consistent estimators for $V_{GMV}$, $R_{GMV}$, $\btheta$, $s$, and $\bet$ are given in \eqref{hV_c}-\eqref{hbet_c}. Similarly, to construct a consistent estimator for $\bL \bQ \bL^\top$, we have that
$$\bL \bQ \bL^\top =\bL \left(\bSigma^{-1}-\frac{\bSigma^{-1}\bi\bi^\top \bSigma^{-1}}{\bi^\top \bSigma^{-1}\bi}\right) \bL^\top
=\bL \bSigma^{-1}\bL^\top - \frac{1}{V_{GMV}}\btheta\btheta^\top.$$
First, $V_{GMV}$ and $\btheta$ are replaced by their consistent estimators $\hat{V}_{GMV;c}$ and $\hat{\btheta}_c$. Second, note that a consistent estimator for $\bl_i^\top \bSigma^{-1} \bl_j$ with deterministic vectors $\bl_i$ and $\bl_j$ satisfying Assumption (\textbf{A1}) is given by $\left(1-p/n\right)\bl_i^\top \hat{\bSigma}^{-1} \bl_j$ (c.f., \citet[Lemma 5.3]{bodnar2019okhrin}). As a result, $\bL \bQ \bL^\top$  is consistently estimated by $\left(1-p/n\right)\bL \hat{\bQ} \bL^\top$ with $\hat{\bQ}$ given in \eqref{sample_bQ} and, hence, $V_{GMV}$, $R_{GMV}$, $\btheta$, $s$, $\bet$, and $\bL \bQ \bL^\top$ with their consistent estimators in \eqref{Omega_gc} and \eqref{Xi_hc}, we obtain consistent estimators for $\bOmega_{\bL,g,c}$ and $\mathbf{\Xi}_{h,c}$ denoted by $\hat{\bOmega}_{\bL,g,c}$ and $\hat{\mathbf{\Xi}}_{h,c}$. For instance, a consistent estimator for the covariance matrix of the estimated weights of the EU portfolio is given by:
 \begin{flalign}
\hat{\bOmega}_{\bL,EU,c}&=\bigg(
	\left(\frac{1-c_n}{\hat{s}_c+c_n} + (\hat{s}_c+c_n)\gamma^{-1} \right)\gamma^{-1} + \hat{V}_{GMV;c}
\bigg) (1-c_n)\bL \hat{\bQ} \bL^\top \\
&+ \gamma^{-2}\Bigg\{\frac{2(1-c_n)c_n^3 }{(\hat{s}_c+c_n)^2}
+  4(1-c_n)c_n\frac{\hat{s}_c(\hat{s}_c+2c_n)}{(\hat{s}_c+c_n)^2}	
+  \frac{2(1-c_n)c_n^2(\hat{s}_c+c_n)^2}{\hat{s}_c^2} -\hat{s}_c^2  \Bigg\} \hat{\bet}_c \hat{\bet}_c^\top,\nonumber
\end{flalign}
where $c_n=p/n$.

The suggested consistent estimators of $\bOmega_{\bL,g,c}$ and $\mathbf{\Xi}_{h,c}$ are then used to derive $(1-\beta)$ asymptotic confidence intervals for the population optimal portfolio weights and their characteristics. In the case of $k$ linear combination of the optimal portfolio weights $\bw_g$ we get
\begin{equation}\label{CI_LWg}
C_{\bL,g;1-\beta}=\left\{\bol{\omega}:~ (n-p)\left(\bL\hat{\bw}_{g;c} - \bL\bw_g\right)^\top \hat{\bOmega}_{\bL,g,c}^{-1}\left(\bL\hat{\bw}_{g;c} - \bL\bw_g\right) \le \chi^2_{k;1-\beta}\right\},
\end{equation}
where $\chi^2_{k;1-\beta}$ denotes the $(1-\beta)$ quantile from the $\chi^2$-distribution with $k$ degrees of freedom.

Finally, using the duality between the interval estimation and the test theory (c.f., \citet{aitchison1964confidence}), a test on the equality of $k$-linear combination of optimal portfolio weights to a preselected vector $\mathbf{r}$ can be derived. Namely, one has to reject the null hypothesis $H_0\; \bL\bw_g=\mathbf{r}$ in favour to the alternative hypothesis $H_0\; \bL\bw_g=\mathbf{r}$ at significance level $\beta$ as soon as $\mathbf{r}$ does not belong to the confidence interval $C_{\bL,g;1-\beta}$, as given in \eqref{CI_LWg}. Similar results are also obtained in the case of optimal portfolio characteristics.

\section{Finite-sample performance and robustness analysis}
The finite-sample performance of the derived high-dimensional asymptotic approximation of the sampling distribution of the estimated optimal portfolio weights is investigated via an extensive Monte Carlo study in this section. Additionally, we study the robustness of the obtained asymptotic distributions to the violation of the assumption of normality used in their derivation. The following four simulation scenarios will be considered in the simulation study:

\begin{figure}[h!]
  \center
\includegraphics[scale=0.5]{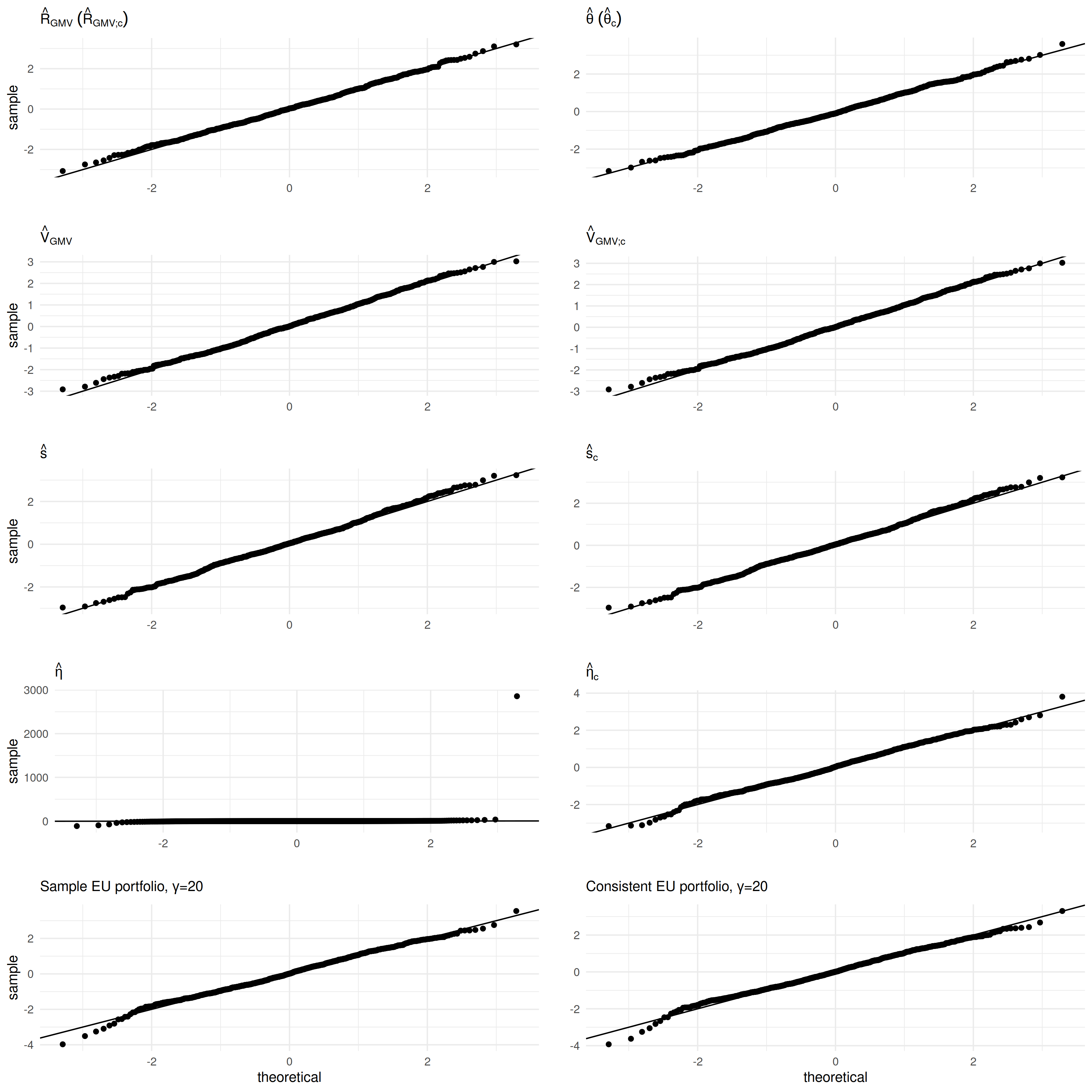}
\caption{QQ-plots of the empirical qunatiles computed for the (standardized) sample estimators (left column) and for the (standardized) consistent estimators (right column) of $V_{GMV}$, $\btheta$, $R_{GMV}$, $s$, $\bet$, and $\bL\bw_{EU}$ in comparison to the quantiles obtained from the corresponding high-dimensional asymptotic distribution. Data of asset returns are generated following \textbf{Scenario 1} with $c=0.5$.}
\label{fig:qq1}
\end{figure}
\begin{figure}[h!]
	\center
	\includegraphics[scale=0.5]{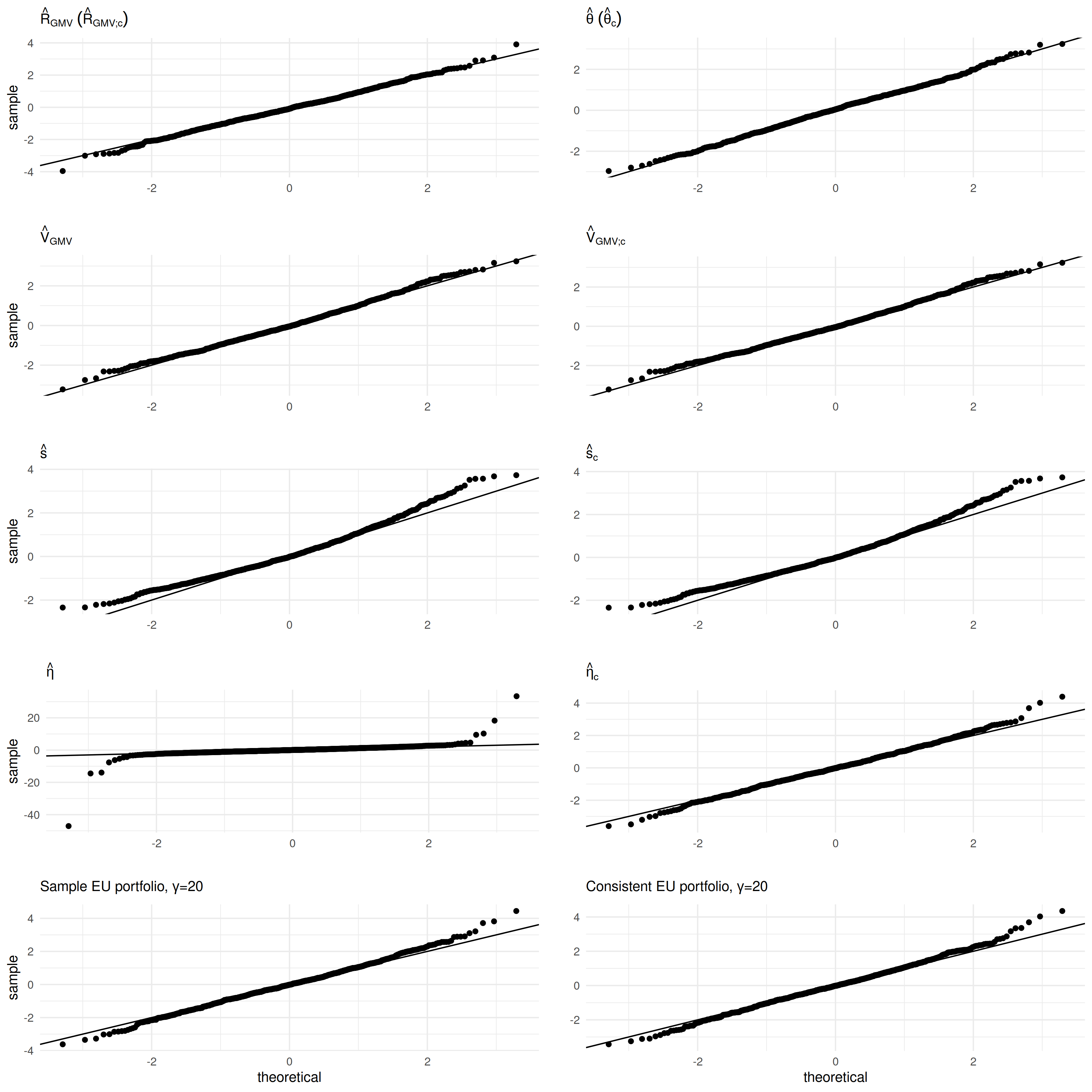}
	\caption{QQ-plots of the empirical qunatiles computed for the (standardized) sample estimators (left column) and for the (standardized) consistent estimators (right column) of $V_{GMV}$, $\btheta$, $R_{GMV}$, $s$, $\bet$, and $\bL\bw_{EU}$ in comparison to the quantiles obtained from the corresponding high-dimensional asymptotic distribution. Data of asset returns are generated following \textbf{Scenario 1} with $c=0.9$.}
	\label{fig:qq5}
\end{figure}

\begin{description}
\item[Scenario 1:] Multivariate normal distribution\\
Sample of asset returns $\bx_1,\bx_2,...,\bx_n$ are generated independently from $N_p(\bmu,\bSigma)$;
\item[Scenario 2:] CAPM model\\
Sample of asset returns $\bx_1,\bx_2,...,\bx_n$ are generated independently from the CAPM model. That is, we sample $n$ observations from $\bx_i = \by_i + \beta z_i$ where $\by_1, ..., \by_n$ are generated independently from $N_p(\bmu,\bSigma)$ and $z_1,...,z_n$ are generated independently from a standard normal distribution;
\item[Scenario 3:] Multivariate $t$-distribution\\
Sample of asset returns $\bx_1,\bx_2,...,\bx_n$ are generated independently from multivariate $t$-distribution with degrees of freedom $d=10$, location parameter $\bmu$, scale matrix $\frac{d-2}{d}\bSigma$. This choice of the scale matrix ensures that the covariance matrix of $\bx_i$ is $\bSigma$;
\item[Scenario 4:] CCC-GARCH model (c.f., \citet{bollerslev1990modelling})\\
The asset returns are assumed to be conditionally normally distributed with $\bx_t \Big| \bSigma_t \sim N_p(\bmu,  \bSigma_t)$. The conditional covariance matrix is specified by $\bSigma_t=\mathbf{D}_t^{1/2} \mathbf{C} \mathbf{D}_t^{1/2}$ with $\mathbf{D}_t=diag(h_{1,t},...,h_{p,t})$ and
$$h_{i,t} = \alpha_{i,0} + \alpha_{1,i} (x_{i,t-1} - \mu_i)^2 + \beta_{1,i}h_{i,t-1},\; \text{for} \; i=1,...,p\text{ and } t=1,...,n,$$
where $\bx_t=(x_{1,t},...,x_{p,t})^\top$, $\bmu=(\mu_1,...,\mu_p)^\top$, $\alpha_{i,0} \ge 0$ and $\alpha_{1,i},\beta_{1,i}>0$ for $i=1,...,p$.
\end{description}

\begin{figure}[h!]
\center
\includegraphics[scale=0.5]{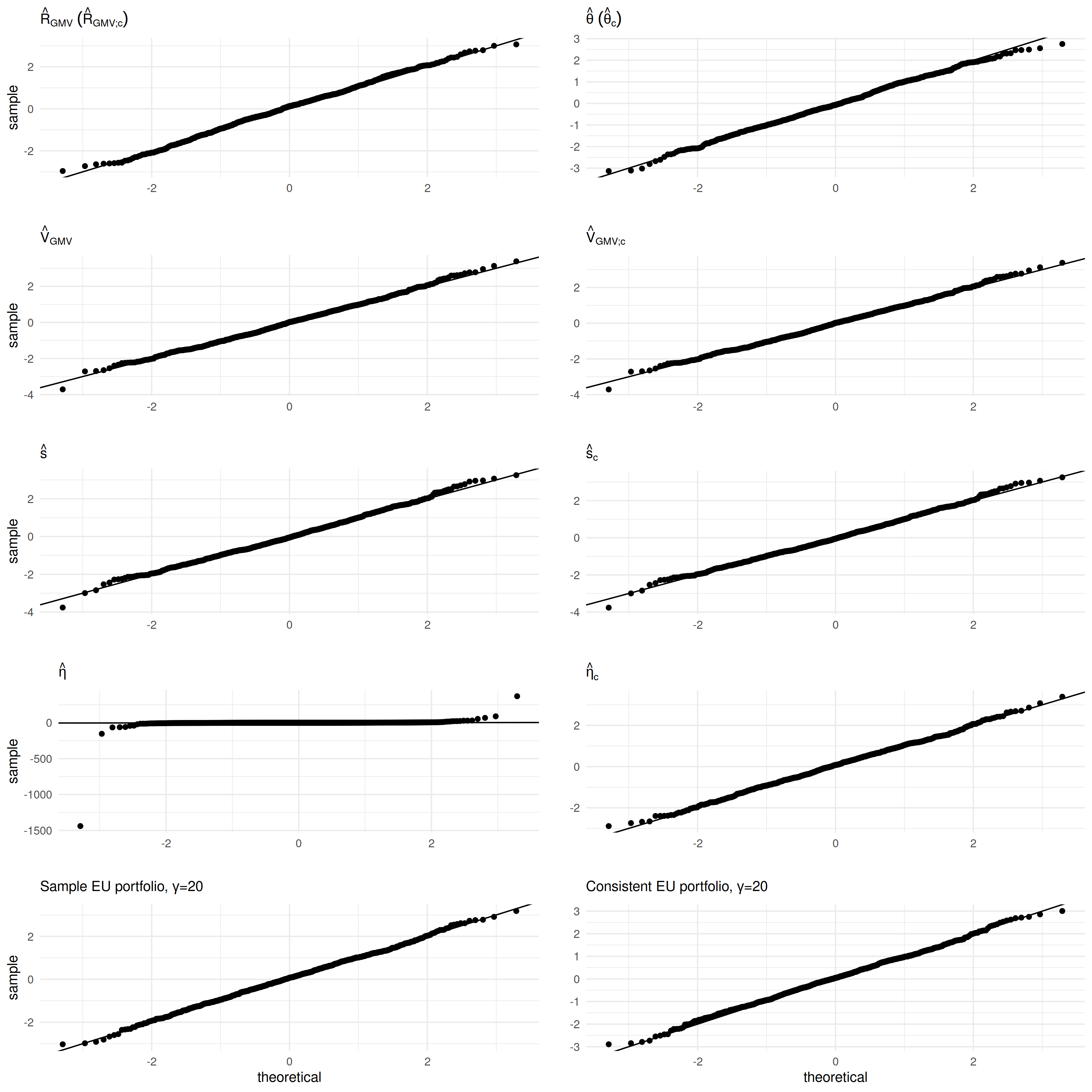}
\caption{QQ-plots of the empirical qunatiles computed for the (standardized) sample estimators (left column) and for the (standardized) consistent estimators (right column) of $V_{GMV}$, $\btheta$, $R_{GMV}$, $s$, $\bet$, and $\bL\bw_{EU}$ in comparison to the quantiles obtained from the corresponding high-dimensional asymptotic distribution. Data of asset returns are generated following \textbf{Scenario 2} with $c=0.5$.}
\label{fig:qq2}
\end{figure}

\begin{figure}[h!]
	\center
	\includegraphics[scale=0.5]{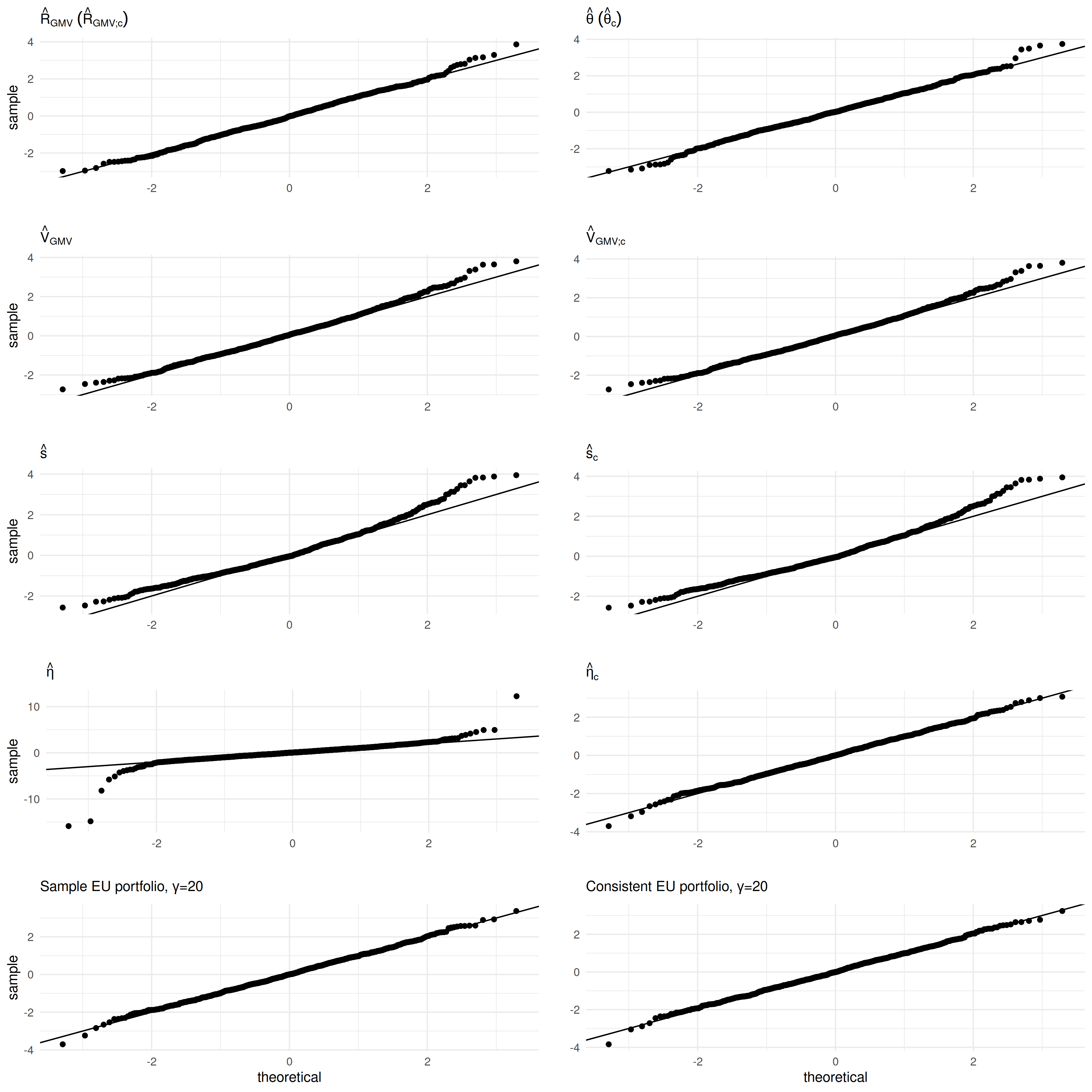}
	\caption{QQ-plots of the empirical qunatiles computed for the (standardized) sample estimators (left column) and for the (standardized) consistent estimators (right column) of $V_{GMV}$, $\btheta$, $R_{GMV}$, $s$, $\bet$, and $\bL\bw_{EU}$ in comparison to the quantiles obtained from the corresponding high-dimensional asymptotic distribution. Data of asset returns are generated following \textbf{Scenario 2} with $c=0.9$.}
	\label{fig:qq6}
\end{figure}

\textbf{Scenario 1} corresponds to the assumption used in the derivation of the theoretical results of the paper. \textbf{Scenario 2} corresponds to a one-factor model of asset returns, which is a popular approach in financial literature. Although, following this model the largest eigenvalue of the covariance matrix $\bSigma +\beta^2 \bi \bi^\top$ is of order $p$ (see, e.g., \citet{fan2013large}), it still fulfills the assumptions used in the derivation of the theoretical findings of Sections 2 and 3. \textbf{Scenario 3} violates the key assumption of normality by allowing heavy tails in the distribution of the asset returns. Finally, \textbf{Scenario 4} is used to investigate the performance of the high-dimensional asymptotic approximation of the sampling distribution of the estimated optimal portfolio weights when the assumption of independence does not hold by introducing non-zero autocorrelation between the squared values of the asset returns.

In order to make the results more flexible, the model parameters are not fixed to some preselected values, but  are randomly simulated. In the considered scenarios the components of $\bmu$ are generated from $U(-0.2,0.2)$, where $U(a,b)$  stands for the uniform distribution on $[a,b]$. The eigenvalues of the covariance matrix $\bSigma$ are fixed, such that 20\% of them are equal to 0.2, 40\% are equal to 1, and 40\% are equal to 5, while its eigenvectors are simulated from the Haar distribution. In \textbf{Scenario 3} we sample $\beta_j \sim U(-0.2, 0.2)$, $j=1,2,...,p$. The coefficients $\alpha_{1,i}, \beta_{1,i}$ of CCC-GARCH model in \textbf{Scenario 4} are generated according to $\alpha_{1,i} \sim U(0, 0.1)$ and $\beta_{1,i}\sim U(0.6, 0.7)$. By such a construction, we always ensure $\alpha_{1,i} + \beta_{1,i} < 1$, i.e., the stationarity condition is fulfilled. Finally, $\alpha_{i,0}$ is chosen such that the unconditional covariance matrix of the CCC-GARCH process is equal to $\bSigma$. Finally, we put $n=1000$ and $c \in \{0.5,0.9\}$ in all four scenarios.

\begin{figure}[h!]
	\center
	\includegraphics[scale=0.5]{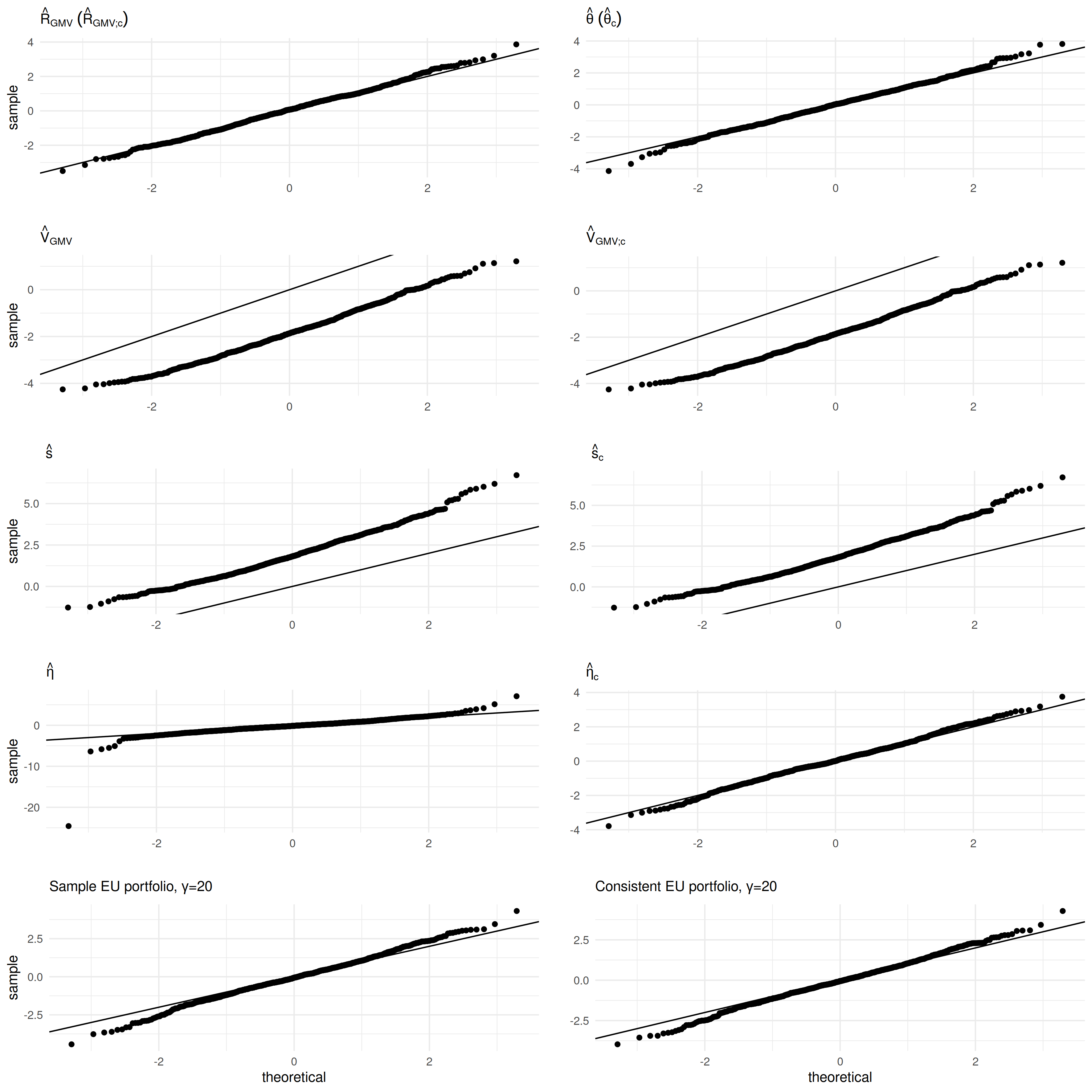}
	\caption{QQ-plots of the empirical qunatiles computed for the (standardized) sample estimators (left column) and for the (standardized) consistent estimators (right column) of $V_{GMV}$, $\btheta$, $R_{GMV}$, $s$, $\bet$, and $\bL\bw_{EU}$ in comparison to the quantiles obtained from the corresponding high-dimensional asymptotic distribution. Data of asset returns are generated following \textbf{Scenario 3} with $c=0.5$.}
	\label{fig:qq3}
\end{figure}

\begin{figure}[h!]
	\includegraphics[width=1\linewidth]{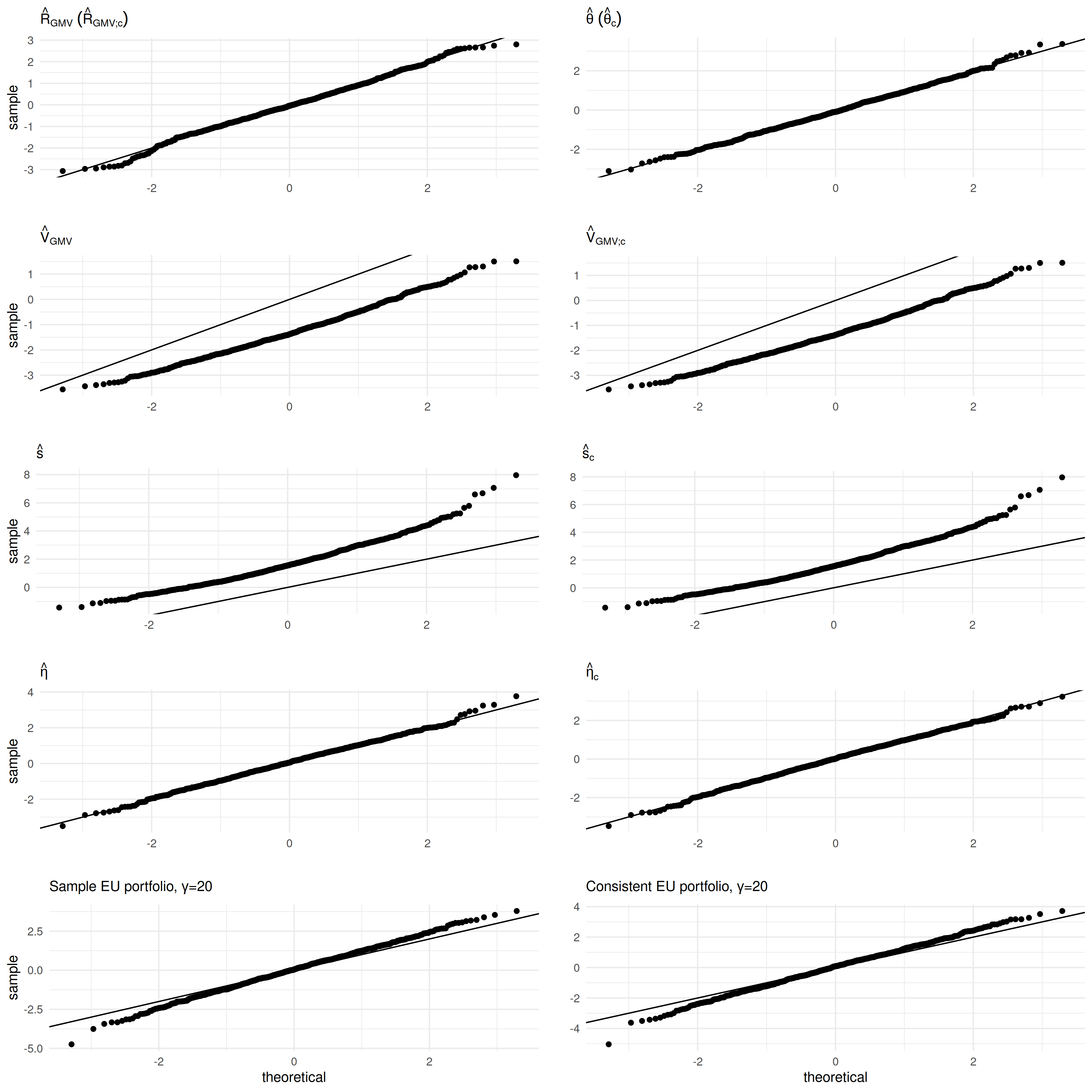}
	\caption{QQ-plots of the empirical qunatiles computed for the (standardized) sample estimators (left column) and for the (standardized) consistent estimators (right column) of $V_{GMV}$, $\btheta$, $R_{GMV}$, $s$, $\bet$, and $\bL\bw_{EU}$ in comparison to the quantiles obtained from the corresponding high-dimensional asymptotic distribution. Data of asset returns are generated following \textbf{Scenario 3} with $c=0.9$.}
	\label{fig:qq7}
\end{figure}

The results of the simulation study are illustrated in the case of five quantities $V_{GMV}$, $\btheta_{GMV}$, $R_{GMV}$, $s$, and $\bet$, and the first weight of the EU portfolio with $\gamma=20$ and $\bL=(1,0,0,...,0)$. We will compute the sample and consistent estimators for these quantities in all scenarios, and compare their finite-sample distributions to the corresponding high-dimensional ones derived in Section 4. In Figures \ref{fig:qq1} to \ref{fig:qq8}, we display QQ-plots for each of the six estimated quantities, both ordinary sample estimator and consistent ones, against the quantiles obtained from their corresponding high-dimensional distribution. On the first row of each figure, we see the estimators for $\hat{R}_{GMV}$ and $\hat{\btheta}$, where the sample and consistent estimators coincide. On the second to fifth row, the first column corresponds to the sample estimators of the considered quantities, while the second column presents the results for the consistent estimators.

In the first and second scenarios, we employ the stochastic representations of Theorems \ref{th2} and \ref{th3}. We sample from the finite-sample distribution of each estimated quantity through $1000$ independent draws $\hat{V}_{GMV}^{(b)}$, $\hat{\btheta}_{GMV}^{(b)}$, $\hat{R}_{GMV}^{(b)}$, $\hat{s}^{(b)}$, $\hat{\bet}^{(b)}$, and $\bL\hat{\bw}_{EU}^{(b)}$ for $b=1,...,1000$. To this end, we note that the application of Theorems \ref{th2} and \ref{th3} provides an efficient way to generate the sample $\hat{V}_{GMV}^{(b)}$, $\hat{\btheta}_{GMV}^{(b)}$, $\hat{R}_{GMV}^{(b)}$, $\hat{s}^{(b)}$, $\hat{\bet}^{(b)}$, and $\bL\hat{\bw}_{EU}^{(b)}$, which also avoids the computation of the inverse sample covariance matrix, which might be an ill-defined object in large dimensions, especially when $c=0.9$. In \textbf{Scenarios 3 and 4}, the stochastic representations derived in Theorems \ref{th2} and \ref{th3} can no longer be used, since the assumptions used for their derivation are no longer fulfilled. In these two cases, we calculate all quantities explicitly for each simulation run. This increases the computational cost immensely, since we need to invert a $p\times p$ matrix for each generated sample.

\begin{figure}[h!]
	\center
	\includegraphics[scale=0.5]{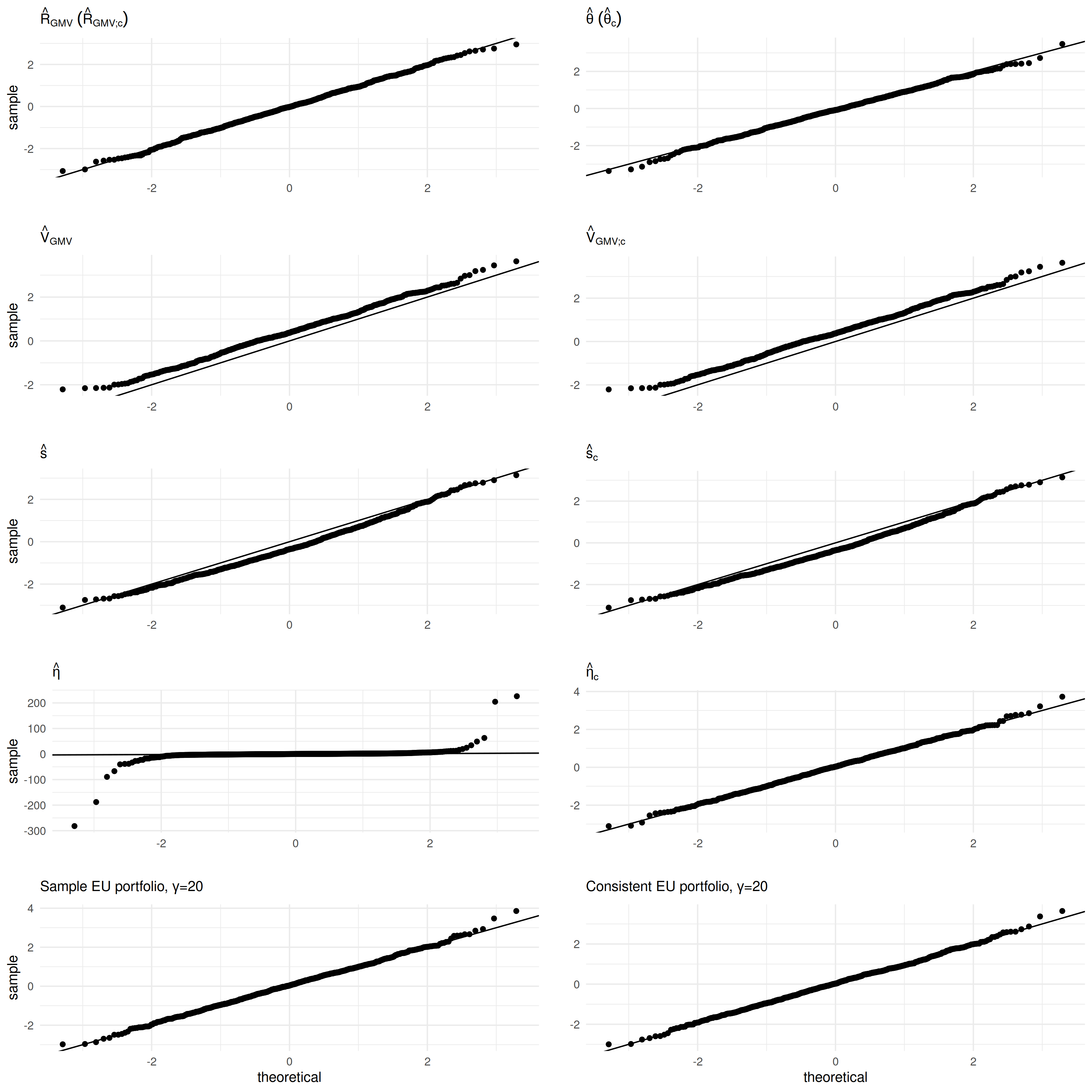}
	\caption{QQ-plots of the empirical qunatiles computed for the (standardized) sample estimators (left column) and for the (standardized) consistent estimators (right column) of $V_{GMV}$, $\btheta$, $R_{GMV}$, $s$, $\bet$, and $\bL\bw_{EU}$ in comparison to the quantiles obtained from the corresponding high-dimensional asymptotic distribution. Data of asset returns are generated following \textbf{Scenario 4} with $c=0.5$.}
	\label{fig:qq4}
\end{figure}

\begin{figure}[h!]
	\includegraphics[width=1\linewidth]{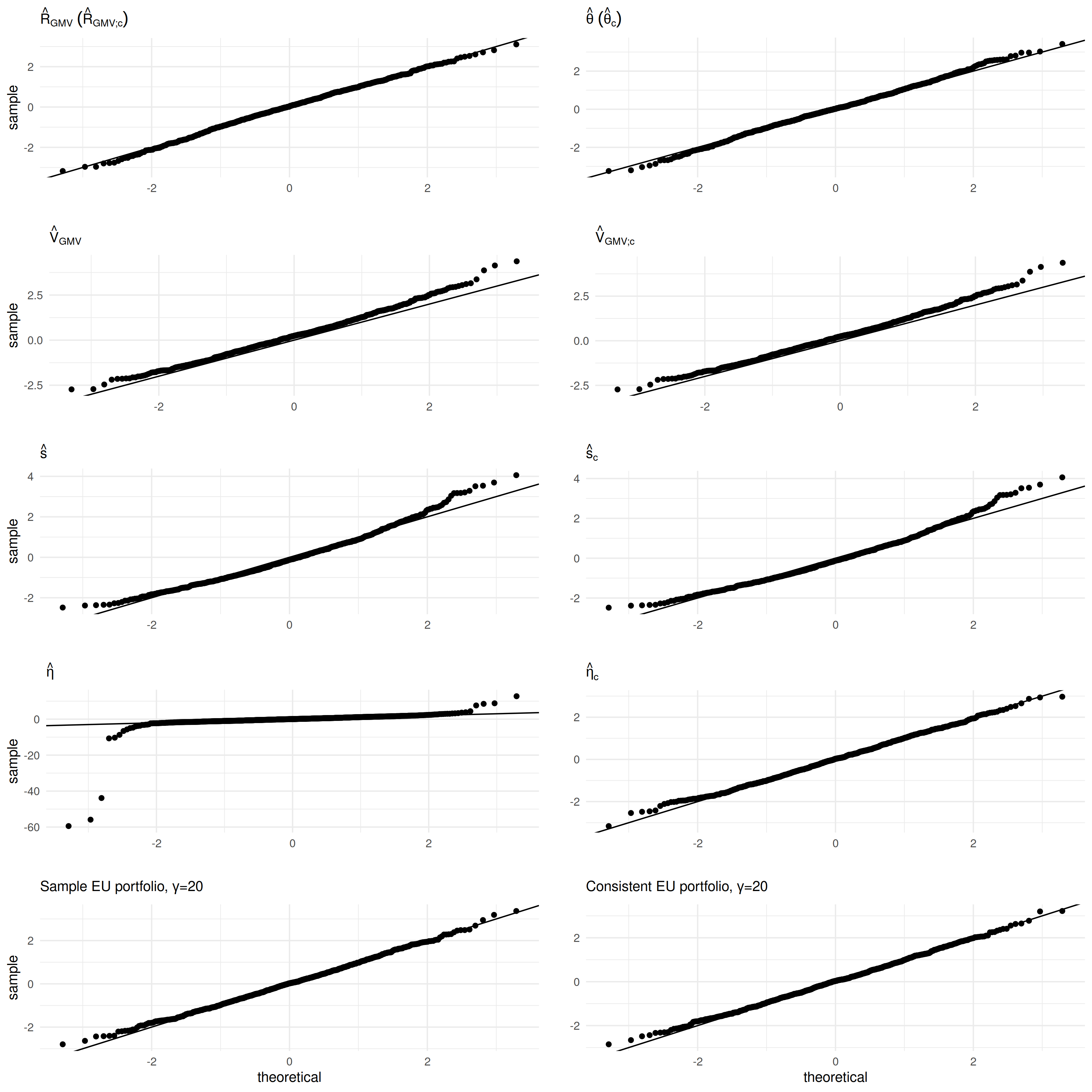}
	\caption{QQ-plots of the empirical qunatiles computed for the (standardized) sample estimators (left column) and for the (standardized) consistent estimators (right column) of $V_{GMV}$, $\btheta$, $R_{GMV}$, $s$, $\bet$, and $\bL\bw_{EU}$ in comparison to the quantiles obtained from the corresponding high-dimensional asymptotic distribution. Data of asset returns are generated following \textbf{Scenario 4} with $c=0.9$.}
	\label{fig:qq8}
\end{figure}

In Figures~\ref{fig:qq1} and \ref{fig:qq5} the results obtained under \textbf{Scenario 1} are depicted. We observe in the figures that the high-dimensional asymptotic distributions provide a good approximation for the moderate value of the concentration ratio $c=0.5$ and its large value $c=0.9$. The approximation seems to be worst off in the context of approximating the distribution of $\hat{s}$ when $c=0.9$, as the tails become much heavier than the approximation seems to be able to account for. Also, there seems to be some instability in the sample estimator for $\hat{\bet}$, as there is one observation which is very large. We do not see the same behaviour for the consistent estimator. Similar findings are also present in Figures~\ref{fig:qq2} and \ref{fig:qq6}, where we simulated from \textbf{Scenario 2}, that is from the CAPM model. The sample estimator for $\hat{\bet}$ seems to have some deviating observations which are no longer observable when the consistent estimator is used.

In Figures \ref{fig:qq3} and \ref{fig:qq7} we display the high-dimensional asymptotic approximations of the sampling distributions of the sample and consistent estimators of $V_{GMV}$, $\btheta$, $R_{GMV}$, $s$, $\bet$, and $\bL\bw_{EU}$. The high-dimensional distributions approximate the corresponding sampling distributions reasonably well when the asset returns are assumed to be multivariate $t$-distributed. There are  some deviations from normality in the sampling distributions of the estimators for $s$ and $V_{GMV}$, regardless of which estimator we use. There are also some observations in the finite-sample distribution of the sample estimator $\hat{\bet}$ that might indicate the presence of heavy tails or skewness, while the same behaviour is not seen in its consistent estimator. There is a small bias for two quantities $s$ and $V_{GMV}$ which could be explained by the influence of heavy tails in the data-generating model on the estimation of the inverse of the high-dimensional covariance matrix. On the other hand, the asymptotic variances seem to be well approximated by the results of Theorems \ref{th5} and \ref{th6}. All other quantities show a good performance despite the violation of the distributional assumption. We also observe the same type of skewness as in \textbf{Scenarios 1 and 2} in the case of $s$, when the asset universe becomes large.

In Figures \ref{fig:qq4} and \ref{fig:qq8} the results of the simulation are shown in the case of \textbf{Scenario 4}. In contrast to the previous three scenarios, here it is assumed that the square asset returns are autocorrelated, which violates the assumption of independence imposed in the derivation of the theoretical results. The empirical findings of both figures also document here that the high-dimensional asymptotic distributions provide a good approximation of the corresponding finite-sample distributions. The only inconsistency we observe is the same type of heavy tail behaviour in the finite-sample distribution of the sample estimator for $\bet$ which disappear in its consistent counterpart. No significant impact of dependence on the performance of the derived high-dimensional distribution is observable in other plots of Figures \ref{fig:qq4} and \ref{fig:qq8}.

\section{Summary}
In this paper we derive the exact sampling distribution of the estimators for a large class of optimal portfolio weights and their estimated characteristics. The results are present in terms of stochastic representations, which provides an easy way to assess the sampling distribution of the estimated optimal portfolio weights. Another important application of the derived stochastic representations is that it presents an efficient way to sample from the corresponding (joint) distribution. The largest computational efficiency comes from the fact that it excludes the inversion of the sample covariance matrix in each simulation run. Furthermore, the derived stochastic representation simplifies the study of the asymptotic properties of the estimated quantities under the high-dimensional asymptotic regime.

In the derivation of the theoretical results we assume that the asset returns are independent and normally distributed. These assumptions seem to be appropriate when asset returns are taken at weekly or lower frequency, while they are not fulfilled for financial data of daily and higher frequency. In our Monte Carlo simulations, we study the finite-sample performance of the obtained asymptotic distributions in comparison to the exact sampling distributions. We investigate the implications of violations to the model assumptions, namely departures from the assumption of normality and independence. We observe a good performance of the asymptotic distributions for finite samples when the data are simulated from the normal distribution. If the assumption of normality is violated, the asymptotic distribution fails to capture some structure in the data. These are biases that appear in the means and variances for a selected few of the estimated quantities. For the rest, the normal approximations still seem to provide a good fit. Assessing the biases in the asymptotic means and in the asymptotic (co)variances of the estimated optimal portfolio weights and their characteristic is an important challenge which will be treated in the consequent paper. Finally, we also find minor impact of the presence of autocorrelation between squared asset returns on the performance of the derived asymptotic distributions.

\section*{Acknowledgement} The authors would like to thank Professor Wang Zhou and two anonymous reviewers for their helpful suggestions. We also want to thank M. Stein who typed parts of this manuscript with considerable technical expertise. This research was partly supported by the Swedish Research Council (VR) via the project ``Bayesian Analysis of Optimal Portfolios and Their Risk Measures''.

\section*{Appendix}

In this section, the proofs of the theoretical results are given. In Lemma \ref{th1} we derive the conditional distribution of $(\hat{V}_{GMV},\hat{\btheta}^\top,\hat{R}_{GMV},\hat{s},\hat{\bet}^\top)^\top$ under the condition $\hbmu=\tbmu$, i.e. the distribution of $(\hat{V}_{GMV},\hat{\btheta}^\top,\tilde{R}_{GMV},\tilde{s},\tilde{\bet}^\top)^\top$.

\begin{lemma}\label{th1}
Under the conditions of Theorem \ref{th2}, the distribution of $(\hat{V}_{GMV},\hat{\btheta}^\top,\tilde{R}_{GMV},\tilde{s},\tilde{\bet}^\top)^\top$ is determined by
\begin{enumerate}[(i)]
\item $\hat{V}_{GMV}$ is independent of $(\hat{\btheta}^\top,\tilde{R}_{GMV},\tilde{s},\tilde{\bet}^\top)^\top$;
\item $(n-1)\frac{\hat{V}_{GMV}}{V_{GMV}}\sim \chi^2_{n-p}$;
\item $\left(\begin{array}{c} \hat{\btheta} \\ \tilde{R}_{GMV}\end{array}\right) \sim
    t_{k+1}\left( n-p+1,\left(\begin{array}{c} \btheta \\ \breve{R}_{GMV}\end{array}\right),
     \frac{V_{GMV}}{n-p+1} \breve{\bG} \right)$,\\
      with $\breve{\bG}=\begin{pmatrix} \bL \bQ \bL^\top & \bL \bQ \tbmu \\ \tbmu^\top \bQ \bL^\top & \tbmu^\top \bQ \tbmu \end{pmatrix}
    =\begin{pmatrix} \bL \bQ \bL^\top & \breve{s} \breve{\bet} \\ \breve{s} \breve{\bet}^\top & \breve{s} \end{pmatrix}
   $
\item $\tilde{s}$ and $\tilde{\bet}$ are conditionally independent given $\hat{\btheta}$ and $\tilde{R}_{GMV}$

\item $(n-1)\frac{\breve{s}}{\tilde{s}}\left(1 + \frac{(\tilde{R}_{GMV} - \breve{R}_{GMV})^2}{V_{GMV}\breve{s}}\right) \sim \chi^2_{n-p+2}$;
\item
$\tilde{\bet}|\hat{\btheta}^\top, \tilde{R}_{GMV} \sim
t_{k}\left(n-p+3,\breve{\bet}+ \mathbf{h},
\frac{(n-p+3)^{-1}\tilde{\bF}}{\breve{s}\left(1+\frac{(\tilde{R}_{GMV} - \breve{R}_{GMV})^2}{V_{GMV}\breve{s}}\right)^2} \right), $
where
\begin{eqnarray*}
\mathbf{h}&=&\left(1 + \frac{(\tilde{R}_{GMV} - \breve{R}_{GMV})^2}{V_{GMV}\breve{s}}\right)^{-1}\frac{(\hat{\btheta}-\btheta-\breve{\bet}(\tilde{R}_{GMV} - \breve{R}_{GMV}))(\tilde{R}_{GMV} - \breve{R}_{GMV})}{V_{GMV}\breve{s}}\\
\tilde{\bF}&=&\left(\bL \bQ \bL^\top-\breve{s}\breve{\bet}\breve{\bet}^\top\right)\left(1+\frac{(\tilde{R}_{GMV} - \breve{R}_{GMV})^2}{V_{GMV}\breve{s}}\right)\\
&+&\frac{\breve{s}}{V_{GMV}}\left(\hbtheta - \btheta-\breve{\bet}(\tilde{R}_{GMV} - \breve{R}_{GMV})\right)\left(\hbtheta - \btheta-\breve{\bet}(\tilde{R}_{GMV} - \breve{R}_{GMV})\right)^\top.
\end{eqnarray*}
\end{enumerate}
\end{lemma}

\begin{proof}[Proof of Lemma~\ref{th1}:]

Under the assumption of independent and normally distributed sample of the asset returns, we get that
\begin{enumerate}[(a)]
\item $\hbmu \sim \mathcal{N}_p\left(\bmu,\dfrac{1}{n}\bSigma\right)$; \label{eqn:hbmu}
\item $(n-1)\hbSigma \sim \mathcal{W}_p(n-1,\bSigma)$ ($p$-dimensional Wishart distribution with $(n-1)$ degrees of freedom and covariance matrix $\bSigma$);
\item $\hbmu$ and $\hbSigma$ are independent.
\end{enumerate}
As a result, the conditional distribution of a random variable defined as a function of $\hbmu$ and $\hbSigma$ given $\hbmu=\tbmu$ is equal to the distribution of a random variable defined by the same function where $\hbmu$ is replaced by $\tbmu$.

Let $\tilde{\bM}=(\bL^\top,\tbmu,\bi)^\top$ and define
$$\tilde{\bH}=\tilde{\bM}\hbSigma^{-1} \tilde{\bM}^\top=
\left(\begin{array}{cc}
\tilde{\bH}_{11}&\tilde{\bH}_{12}\\
\tilde{\bH}_{21}&\tilde{\bH}_{22}\\
\end{array}\right)$$
with
$$\tilde{\bH}_{11}=
\left(\begin{array}{cc}
\bL\hbSigma^{-1}\bL^\top&\bL\hbSigma^{-1}\tbmu\\
\tbmu^\top\hbSigma^{-1}\bL^\top&\tbmu^\top\hbSigma^{-1}\tbmu\\
\end{array}\right),
~ \tilde{\bH}_{12}= \left(\begin{array}{c}
\bL\hbSigma^{-1}\bi\\
\tbmu^\top\hbSigma^{-1}\bi\\
\end{array}\right),~ \tilde{\bH}_{21}=\tilde{\bH}_{12}^\top, ~ \tilde{\bH}_{22}=\bi^\top\hbSigma^{-1}\bi$$
and
$$\bH=\tilde{\bM}\bSigma^{-1} \tilde{\bM}^\top=
\left(\begin{array}{cc}
\bH_{11}&\bH_{12}\\
\bH_{21}&\bH_{22}\\
\end{array}\right)$$
with
$$\bH_{11}=
\left(\begin{array}{cc}
\bL\bSigma^{-1}\bL^\top&\bL\bSigma^{-1}\tbmu\\
\tbmu^\top\bSigma^{-1}\bL^\top&\tbmu^\top\bSigma^{-1}\tbmu\\
\end{array}\right),~
\bH_{12}= \left(\begin{array}{c}
\bL\bSigma^{-1}\bi\\
\tbmu^\top\bSigma^{-1}\bi\\
\end{array}\right), ~ \bH_{21}=\bH_{12}^\top, ~ \bH_{22}=\bi^\top\bSigma^{-1}\bi.$$

Also, let
\begin{equation}\label{tbG}
\tilde{\bG}=\tilde{\bH}_{11}-\frac{\tilde{\bH}_{12}\tilde{\bH}_{21}}{\tilde{\bH}_{22}}
=\left(\begin{array}{c}\bL\\ \tbmu^\top \end{array}\right)\hat{\bQ}\left(\begin{array}{cc}\bL^\top&\tbmu \end{array}\right)
=\left(\begin{array}{cc}\bL\hat{\bQ}\bL^\top & \bL\hat{\bQ}\tbmu\\
\tbmu^\top \hat{\bQ}\bL^\top & \tbmu^\top\hat{\bQ}\tbmu \end{array}\right)
=\left(\begin{array}{cc}\tilde{\bG}_{11} & \tilde{\bG}_{12}\\
\tilde{\bG}_{21} & \tilde{\bG}_{22}\end{array}\right)
\end{equation}
and
\begin{equation}\label{bG}
\bG=\bH_{11}-\frac{\bH_{12}\bH_{21}}{\bH_{22}}=\left(\begin{array}{c}\bL\\ \tbmu^\top \end{array}\right)\bQ\left(\begin{array}{cc}\bL^\top&\tbmu \end{array}\right)
=\left(\begin{array}{cc}\bL\bQ\bL^\top & \bL\bQ\tbmu\\
\tbmu^\top \bQ\bL^\top & \tbmu^\top\bQ\tbmu \end{array}\right)
=\left(\begin{array}{cc}\bG_{11} & \bG_{12}\\
\bG_{21} & \bG_{22}\end{array}\right)
\end{equation}
with $\tilde{\bG}_{22}=\tbmu^\top\hat{\bQ}\tbmu$ and $\bG_{22}=\tbmu^\top\bQ\tbmu$.

In using the definitions of $\tilde{\bH}$ and $\tilde{\bG}$, we get
\[\hat{V}_{GMV}=\frac{1}{\tilde{\bH}_{22}}, ~~\left(\begin{array}{c} \hat{\btheta}\\\tilde{R}_{GMV} \end{array}\right)=\frac{\tilde{\bH}_{12}}{\tilde{\bH}_{22}},~~
\tilde{s}=\tilde{\bG}_{22}, ~~ \tilde{\bet}= \frac{\tilde{\bG}_{12}}{\tilde{\bG}_{22}}.
\]
Moreover, from \citet[Theorem 3.2.11]{muirhead2009aspects} we get $(n-1)\tilde{\bH}^{-1} \sim \mathcal{W}_{k+2}(n-p+k+1,\bH^{-1})$ and, consequently, (see, \citet[Theorem 3.4.1] {gupta2000matrix} $(n-1)^{-1}\tilde{\bH} \sim \mathcal{W}_{k+2}^{-1}(n-p+2k+4,\bH)$ ($(k+2)$-dimensional inverse Wishart distribution with $n-p+2k+4$ degrees of freedom and parameter matrix $\bH$). The application of Theorem 3 in \citet{bodnar2008properties} leads to
\begin{enumerate}[(i)]
\item $\tilde{\bH}_{22}$ is independent of ${\tilde{\bH}_{12}}/{\tilde{\bH}_{22}}$ and $\tilde{\bG}$ and, consequently,
    \[\hat{V}_{GMV}~~ \text{is independent of}~~(\hat{\btheta}^\top,\tilde{R}_{GMV},\tilde{s},\tilde{\bet}^\top)^\top.\]
\item We get that $(n-1)^{-1}\tilde{\bH}_{22}\sim \mathcal{W}_{1}^{-1}(n-p+2,\bH_{22})$. Hence,
\begin{equation}\label{H_22}
(n-1)\frac{\bi^\top\bSigma^{-1}\bi}{\bi^\top\hbSigma^{-1}\bi}=(n-1)\frac{\hat{V}_{GMV}}{V_{GMV}}\sim \chi^2_{n-p}\,;
\end{equation}
\item Let
$\Gamma_l\Big(\frac{m}{2}\Big)=\pi^{l(l-1)/4}\prod_{i=1}^l \Gamma\Big(\frac{m-i+1}{2}\Big)$
be the multivariate gamma function. Then, the density of $\tilde{\bH}_{12}/\tilde{\bH}_{22}=\begin{pmatrix} \hat{\btheta}^\top & \tilde{R}_{GMV} \end{pmatrix}^\top$ is given by
\begin{eqnarray}
f(\by) &=& \frac{|\bG|^{-\frac{1}{2}}|\bH_{22}|^{\frac{(k+1)}{2}}}
{\pi^{\frac{k+1}{2}}}\frac{\Gamma_{k+1}(\frac{n-p+k+2}{2})}{\Gamma_{k+1}(\frac{n-p+k+1}{2})}|\bI+\bG^{-1}(\by-\bH_{12}/\bH_{22})\bH_{22}(\by-\bH_{12}/\bH_{22})^\top|^{-\frac{n-p+k+2}{2}}\,\nonumber\\
&=&\frac{|\bG/\bH_{22}|^{-\frac{1}{2}}}
{\pi^{\frac{k+1}{2}}}\frac{\Gamma_{k+1}(\frac{n-p+k+2}{2})}{\Gamma_{k+1}(\frac{n-p+k+1}{2})}\left(1+\bH_{22}(\by-\bH_{12}/\bH_{22})^\top\bG^{-1}(\by-\bH_{12}/\bH_{22})\right)^{-\frac{n-p+k+2}{2}} \label{eqn:dens}
\end{eqnarray}
where the last equality is obtained by the use of the Sylvester determinant identity. The density presented in \eqref{eqn:dens} corresponds to a $(k+1)$-dimensional $t$ distribution with $(n-p+1)$ degrees of freedom, location parameter $\bH_{12}/\bH_{22} = \begin{pmatrix} \btheta^\top & \breve{R}_{GMV} \end{pmatrix}^\top$ and scale matrix $\frac{V_{GMV}}{n-p+1}\bG$.
\end{enumerate}

In the proof of parts (iv)-(vi) we use the following result (see Theorem 3.f of \citet{bodnar2008properties})
\begin{equation*}
(n-1)^{-1}\tilde{\bG}| \hat{\btheta}^\top, \tilde{R}_{GMV} \sim \mathcal{W}_{k+1}^{-1}\left(n-p+2k+4,\tilde{\bB}\right),
\end{equation*}
where
\[
\tilde{\bB}=\bG+
\frac{1}{V_{GMV}}
\left(\begin{array}{c}
\hbtheta - \btheta \\ \tilde{R}_{GMV} - \breve{R}_{GMV}
\end{array}\right)
\left(\begin{array}{c}
\hbtheta - \btheta \\ \tilde{R}_{GMV} - \breve{R}_{GMV}
\end{array}\right)^\top
=\left(\begin{array}{cc}\tilde{\bB}_{11} & \tilde{\bB}_{12}\\
\tilde{\bB}_{21} & \tilde{\bB}_{22}\end{array}\right)
\]
with $\tilde{\bB}_{22}=\bG_{22}+\frac{(\tilde{R}_{GMV} - \breve{R}_{GMV})^2}{V_{GMV}}$

Hence,
\begin{itemize}
\item[(iv)] $\tilde{s}=\tilde{\bG}_{22}$ and $\tilde{\bet}={\tilde{\bG}_{12}}/{\tilde{\bG}_{22}}$ are conditionally independent given $\hat{\btheta}^\top$ and $\tilde{R}_{GMV}$.

\item[(v)] It holds that $(n-1)^{-1}\tilde{\bG}_{22}|\hat{\btheta}^\top, \tilde{R}_{GMV}\sim \mathcal{W}_{1}^{-1}\left(n-p+4,\tilde{\bB}_{22}\right)$. Hence,
\begin{equation}\label{G_22}
(n-1)\frac{\breve{s} + (\tilde{R}_{GMV} - \breve{R}_{GMV})^2/V_{GMV}}{\tilde{s}} \sim \chi^2_{n-p+2}\,.
\end{equation}

\item[(vi)] Finally, similarly to the proof of part (iii), we get
\[
\tilde{\bet}|\hat{\btheta}^\top, \tilde{R}_{GMV} \sim
t_{k}\left(n-p+3,\frac{\tilde{\bB}_{12}}{\tilde{\bB}_{22}},
\frac{1}{n-p+3}\frac{\tilde{\bB}_{11}\tilde{\bB}_{22}-\tilde{\bB}_{12}\tilde{\bB}_{21}}{\tilde{\bB}_{22}^2} \right), \]
where
{\small
\begin{eqnarray*}
&&\tilde{\bB}_{11}\tilde{\bB}_{22}-\tilde{\bB}_{12}\tilde{\bB}_{21}\\
&=&
\left(\bG_{11}+\frac{1}{V_{GMV}}(\hbtheta - \btheta)(\hbtheta - \btheta)^\top\right)
    \left(\bG_{22}+\frac{(\tilde{R}_{GMV} - \breve{R}_{GMV})^2}{V_{GMV}}\right)\\
  &-&\left(\bG_{12}+\frac{\tilde{R}_{GMV} - \breve{R}_{GMV}}{V_{GMV}}(\hbtheta - \btheta)\right)
      \left(\bG_{12}+\frac{\tilde{R}_{GMV} - \breve{R}_{GMV}}{V_{GMV}}(\hbtheta - \btheta)\right)^\top\\
  &=& \bG_{11}\bG_{22}-\bG_{12}\bG_{21}+\frac{\bG_{22}}{V_{GMV}}(\hbtheta - \btheta)(\hbtheta - \btheta)^\top\\
  &+& \frac{(\tilde{R}_{GMV} - \breve{R}_{GMV})^2}{V_{GMV}}\bG_{11}-\frac{\tilde{R}_{GMV} - \breve{R}_{GMV}}{V_{GMV}}\left((\hbtheta - \btheta)\bG_{12}^\top+\bG_{12}(\hbtheta - \btheta)^\top\right)\\
  &=& \left(\bG_{11}-\frac{\bG_{12}\bG_{21}}{\bG_{22}}\right)\left(\bG_{22}+\frac{(\tilde{R}_{GMV} - \breve{R}_{GMV})^2}{V_{GMV}}\right)\\
  &+&\frac{\bG_{22}}{V_{GMV}}\left(\hbtheta - \btheta-\frac{\bG_{12}}{\bG_{22}}(\tilde{R}_{GMV} - \breve{R}_{GMV})\right)\left(\hbtheta - \btheta-\frac{\bG_{12}}{\bG_{22}}(\tilde{R}_{GMV} - \breve{R}_{GMV})\right)^\top.
\end{eqnarray*}
}
\end{itemize}
\end{proof}

\begin{proof}[Proof of Theorem~\ref{th2}:]
From Theorem \ref{th1}.ii we get
\begin{equation}\label{ap_th2_eq1}
\hat{V}_{GMV}\stackrel{d}{=}\frac{V_{GMV}}{n-1}\xi_1
\end{equation}
where $\xi_1\sim \chi^2_{n-p}$. Moreover, Theorem~\ref{th1}.iii implies that $\hat{\btheta}$ and $\tilde{R}_{GMV}$ are jointly multivariate $t$-distributed and, hence, it holds that (see, e.g., \citet{ding2016conditional})
$\tilde{R}_{GMV} \sim t\left(n-p+1,\breve{R}_{GMV},\frac{V_{GMV}\breve{s}}{n-p+1}\right)$
and
{\footnotesize
\begin{eqnarray*}
 \hat{\btheta} | \tilde{R}_{GMV}&\sim& t_k\Bigg(n-p+2,\btheta+\breve{\bet}(\tilde{R}_{GMV} - \breve{R}_{GMV}), \frac{n-p+1+(n-p+1)(\tilde{R}_{GMV} - \breve{R}_{GMV})^2/(V_{GMV}\breve{s})}{n-p+2}\frac{V_{GMV}}{n-p+1}\left(\bL \bQ \bL^\top-\breve{s} \breve{\bet}\breve{\bet}^\top\right)\Bigg)\\
 &=&t_k\Bigg(n-p+2,\btheta+\breve{\bet}(\tilde{R}_{GMV} - \breve{R}_{GMV}),\frac{V_{GMV}}{n-p+2}\left(1+\frac{(\tilde{R}_{GMV} - \breve{R}_{GMV})^2}{V_{GMV}\breve{s}}\right)\left(\bL \bQ \bL^\top-\breve{s} \breve{\bet}\breve{\bet}^\top\right)\Bigg)
\end{eqnarray*}}
As a result, we get
\begin{equation}\label{ap_th2_eq2}
\hat{R}_{GMV}\stackrel{d}{=}\frac{\bi^\top\bSigma^{-1}\hbmu}{\bi^\top \bSigma^{-1}\bi}+\frac{\sqrt{V_{GMV}}\sqrt{\hbmu^\top \bQ \hbmu}}{\sqrt{n-p+1}}t_1
\end{equation}
and
\begin{eqnarray}
\hat{\btheta}&\stackrel{d}{=}&\btheta+\sqrt{V_{GMV}}
\frac{t_1}{\sqrt{n-p+1}}\frac{\bL\bQ \hbmu}{\sqrt{\hbmu^\top \bQ \hbmu}}+
\sqrt{1+\frac{t_1^2}{n-p+1}}\frac{\sqrt{V_{GMV}}}{\sqrt{n-p+2}}\left(\bL \bQ \bL^\top-\frac{\bL\bQ \hbmu\hbmu^\top \bQ \bL^\top}{\hbmu^\top \bQ \hbmu}\right)^{1/2}\bt_2\nonumber\\
&=&\btheta+\sqrt{V_{GMV}}
\Bigg(\frac{\bL\bQ \hbmu}{\sqrt{\hbmu^\top \bQ \hbmu}}\frac{t_1}{\sqrt{n-p+1}}
+
\left(\bL \bQ \bL^\top-\frac{\bL\bQ \hbmu\hbmu^\top \bQ \bL^\top}{\hbmu^\top \bQ \hbmu}\right)^{1/2}
\sqrt{1+\frac{t_1^2}{n-p+1}}\frac{\bt_2}{\sqrt{n-p+2}}\Bigg)\label{ap_th2_eq3}
\end{eqnarray}
where $t_1 \sim t(n-p+1)$, $\bt_2\sim t_k(n-p+2)$ are independent and also they are independent of $\xi_1$.

Similarly, the application of Theorem \ref{th1}.v leads to
\begin{equation}\label{ap_th2_eq4}
\hat{s}\stackrel{d}{=}(n-1)\left(1+\frac{t_1^2}{n-p+1}\right)\frac{\hbmu^\top \bQ \hbmu}{\xi_2}, ~\text{where $\xi_2\sim \chi^2_{n-p+2}$ and is independent of $t_1$, $\bt_2$, and $\xi_1$.}
\end{equation}

Finally, the application of Theorem \ref{th1}.vi leads to
{\footnotesize
\begin{eqnarray}
\hat{\bet}&\stackrel{d}{=}&\frac{\bL\bQ \hbmu}{\hbmu^\top \bQ \hbmu}
+\frac{\sqrt{1+\frac{t_1^2}{n-p+1}}\left(\bL \bQ \bL^\top-\frac{\bL\bQ \hbmu\hbmu^\top \bQ \bL^\top}{\hbmu^\top \bQ \hbmu}\right)^{1/2}\frac{\bt_2}{\sqrt{n-p+2}} \frac{1}{\sqrt{\hbmu^\top \bQ \hbmu}}\frac{t_1}{\sqrt{n-p+1}}
                            }{1+\frac{t_1^2}{n-p+1}}\nonumber\\
  &+&\frac{1}{\sqrt{\hbmu^\top \bQ \hbmu}} \frac{1}{1+\frac{t_1^2}{n-p+1}}
 \Bigg(\left(\bL \bQ \bL^\top-\frac{\bL\bQ \hbmu\hbmu^\top \bQ \bL^\top}{\hbmu^\top \bQ \hbmu}\right)\left(1+\frac{t_1^2}{n-p+1}\right)\nonumber\\
&+&\frac{\hbmu^\top \bQ \hbmu}{V_{GMV}}\left(1+\frac{t_1^2}{n-p+1}\right)\frac{V_{GMV}}{n-p+2}
\left(\bL \bQ \bL^\top-\frac{\bL\bQ \hbmu\hbmu^\top \bQ \bL^\top}{\hbmu^\top \bQ \hbmu}\right)^{1/2}\bt_2\bt_2^\top \left(\left(\bL \bQ \bL^\top-\frac{\bL\bQ \hbmu\hbmu^\top \bQ \bL^\top}{\hbmu^\top \bQ \hbmu}\right)^{1/2}\right)^\top
\Bigg)^{1/2}\frac{\bt_3}{\sqrt{n-p+3}}\nonumber\\
&=&\frac{\bL\bQ \hbmu}{\hbmu^\top \bQ \hbmu}
+\frac{1}{\sqrt{\hbmu^\top \bQ \hbmu}}\left(\bL \bQ \bL^\top-\frac{\bL\bQ \hbmu\hbmu^\top \bQ \bL^\top}
    {\hbmu^\top \bQ \hbmu}\right)^{1/2}\nonumber\\
  &\times&\left(\frac{1}{\sqrt{1+\frac{t_1^2}{n-p+1}}}\frac{\bt_2}{\sqrt{n-p+2}}\frac{t_1}{\sqrt{n-p+1}}
+ \left(\bI_k+\hbmu^\top \bQ \hbmu \frac{\bt_2\bt_2^\top}{n-p+2}\right)^{1/2}\frac{\bt_3}{\sqrt{n-p+3}} \right)\label{ap_th2_eq5}
\end{eqnarray}
}
where $\bt_3 \sim t_k(n-p+3)$ and is independent of $t_1$ and $\bt_2$. Moreover, due to Theorem \ref{th1}.i and \ref{th1}.iv we get that $\xi_1$, $\xi_2$, $t_1$, $\bt_2$, and $\bt_3$ are mutually independent.

Next, we derive stochastic representations for the linear and quadratic forms in $\hbmu$, namely of $\bi^\top\bSigma^{-1}\hbmu$, $\bL\bQ \hbmu$ and $\hbmu^\top \bQ \hbmu$ which are present in the derived above stochastic representations. Let $\bP= \bQ\bL^\top\left(\bL \bQ \bL^\top\right)^{-1/2}$ and
$\bA=\bQ-\bP\bP^\top =\bQ- \bQ \bL^\top \left(\bL \bQ \bL^\top\right)^{-1}\bL \bQ$. Then
\begin{equation}\label{ap_th2_eq6}
\hbmu^\top \bQ \hbmu=\hbmu^\top\bA \hbmu+(\bP^\top \hbmu)^\top (\bP^\top \hbmu).
\end{equation}
Moreover, the equality $\bi^\top \bQ=\mathbf{0}^\top$ implies
$
\left(\begin{array}{c}
\bi^\top \bSigma^{-1}\\
\bP^\top
\end{array}
\right)\bSigma \bA =\left(\begin{array}{c}
\bi^\top \bA\\
\bP^\top \bSigma \bA \end{array}
\right)=
\left(\begin{array}{c}
\mathbf{0}^\top\\
\bP^\top- \bP^\top\end{array}
\right)=\mathbf{O}
$
and, consequently, we get from Theorem 5.5.1 in \citet{mathai1992quadratic} that $\hbmu^\top\bA \hbmu$ is independent of $\bi^\top\bSigma^{-1}\hbmu$ and $\bP^\top \hbmu$, while Corollary 5.1.3a in \citet{mathai1992quadratic} implies that
\begin{equation}\label{ap_th2_eq7}
n \hbmu^\top\bA \hbmu\stackrel{d}{=} \xi_3, ~~\text{where $\xi_3\sim \chi^2_{p-k-1;n\bmu^\top\bA \bmu}$}.
\end{equation}

Finally, the identity $\bi^\top \bSigma^{-1}\bSigma\bP=\mathbf{0}$ ensures that $\bi^\top\bSigma^{-1}\hbmu$ and $\bP^\top \hbmu$ are independent (c.f., \citet[Chapter 2.2]{rencher1998multivariate}) with
\begin{eqnarray}\label{ap_th2_eq8}
\bi^\top \bSigma^{-1}\hbmu&\stackrel{d}{=}& \bi^\top \bSigma^{-1}\bmu+\sqrt{\bi^\top \bSigma^{-1}\bi}\frac{z_1}{\sqrt{n}}
=\frac{R_{GMV}}{V_{GMV}}+\frac{1}{\sqrt{V_{GMV}}}\frac{z_1}{\sqrt{n}}\\
\bP^\top \hbmu&\stackrel{d}{=}& \bP^\top \bmu+\left(\bP^\top\bSigma \bP\right)^{1/2}\frac{\tilde{\bz}_2}{\sqrt{n}}
= \left(\bL \bQ \bL^\top\right)^{-1/2} s\bet+ \frac{\bz_2}{\sqrt{n}}\label{ap_th2_eq9}
\end{eqnarray}
where $z_1\sim \mathcal{N}(0,1)$ and $\tilde{\bz}_2\sim \mathcal{N}_k(\mathbf{0},\bI_k)$ are independent. Inserting \eqref{ap_th2_eq6} -- \eqref{ap_th2_eq9} in \eqref{ap_th2_eq1} -- \eqref{ap_th2_eq5} and performing some algebra, we get the statement of the theorem.
\end{proof}


\begin{proof}[Proof of Theorem~\ref{th3}:]
The statement of the theorem follows directly from the results of Theorem \ref{th2}.
\end{proof}

\begin{proof}[Proof of Theorem~\ref{th4}:]
The mutual independence of $\xi$, $\psi$, and $z$ follows from Theorem~\ref{th2}, while Theorem \ref{th2}.i provides the stochastic representation for $\hat{V}_{GMV}$.

Next, we derive the joint stochastic representation for $\hat{R}_{GMV}$ and $\hat{s}$. Let $\tilde\xi_2=\xi_2^{-1}$, then the distribution of $(\hat{R}_{GMV}, \hat{s}, t_1, f)$ is obtained as a transformation of $(z_1, \tilde{\xi}_2, t_1, f)$ with the Jacobian matrix given by
$$
\mathbf{J} =
  \begin{pmatrix}
    \frac{\sqrt{V_{GMV}}}{\sqrt{n}} & 0 & \frac{\sqrt{f}\sqrt{V_{GMV}}}{\sqrt{n-p+1}} & \frac{1}{2}\frac{\sqrt{V_{GMV}}t_1}{\sqrt{n-p+1}\sqrt{f}} \\
    0 & (n-1)\left(1+\frac{t^2_1}{n-p+1}\right)f &\frac{2(n-1)}{n-p+1}ft_1\tilde{\xi}_2 & (n-1)\left(1+\frac{t_1^2}{n-p+1}\right)\tilde{\xi}_2 \\
    0 & 0 & 1 & 0 \\
    0 & 0 & 0 & 1 \\
  \end{pmatrix}
$$
which implies that $|\mathbf{J}|=\frac{(n-1)}{\sqrt{n}} \sqrt{V_{GMV}}\left(1+\frac{t^2_1}{n-p+1}\right)f$.

Let $d_f(\cdot)$ denote the marginal density of the distribution of $f$. Ignoring the normalizing constants, we get the joint density of $(\hat{R}_{GMV}, \hat{s}, t_1, f)$ expressed as
{\small
\begin{flalign*}
& d(\hat{R}_{GMV}, \hat{s}, t_1, f)
 \propto
  \exp\left\{-\frac{n}{2} \frac{\left(\hat{R}_{GMV}-R_{GMV}-\sqrt{f}\frac{t_1\sqrt{V_{GMV}}}{\sqrt{n-p+1}} \right)^2}{V_{GMV}} \right\}
  \left(\frac{(n-1)f}{\hat{s}}\left( 1+\frac{t^2_1}{n-p+1} \right)\right)^{\frac{n-p+2}{2} +1 }
  \nonumber \\
  &\times
    \exp\left\{-\frac{(n-1)f}{2\hat{s}} \left( 1+\frac{t^2_1}{n-p+1}\right)\right\}
  \left(1+\frac{t^2_1}{n-p+1}\right)^{-\frac{n-p+2}{2}} \left(f\left(1+\frac{t^2_1}{n-p+1}\right)\right)^{-1} d_f(f)  \nonumber \\
  &\propto  \left( \frac{f}{\hat{s}} \right)^{\frac{n-p+2}{2} + 1} \frac{1}{f} \exp\Bigg\{-\frac{n}{2}\frac{\left(\hat{R}_{GMV}-R_{GMV}\right)^2}{V_{GMV}}
  +\frac{n\left(\hat{R}_{GMV}-R_{GMV}\right)\sqrt{f}\frac{t_1}{\sqrt{n-p+1}}}{\sqrt{V_{GMV}}}\\
  &-\frac{(n-1)f}{2\hat{s}} -\frac{1}{2} \left(n+\frac{n-1}{\hat{s}}\right)\frac{ft^2_1}{n-p+1}\Bigg\} d_f(f). \nonumber
\end{flalign*}}
We now notice that
{\small
\begin{align*}
  &\exp\left\{\frac{n\left(\hat{R}_{GMV}-R_{GMV}\right)\sqrt{f}\frac{t_1}{\sqrt{n-p+1}}}{\sqrt{V_{GMV}}} - \frac{(n\hat{s}+(n-1))f}{2\hat{s}(n-p+1)}t^2_1\right\}\\
  &=\exp\left\{ - \frac{(n\hat{s}+(n-1))f }{2\hat{s}(n-p+1)}\left(t_1 - \frac{n^2\hat{s}\sqrt{n-p+1}\left(\hat{R}_{GMV}-R_{GMV}\right)}{ \sqrt{V_{GMV}f}(n\hat{s}+(n-1))}\right)^2 \right\}  \exp\left\{ \frac{n^2\hat{s} \left(\hat{R}_{GMV}-R_{GMV}\right)^2}{2V_{GMV}(n\hat{s}+(n-1)) } \right\},
\end{align*}}
where the first factor is the kernel of a normal distribution. Hence,
{\small
\begin{flalign}
 d(\hat{R}_{GMV}, \hat{s})
 &= \int_{\mathbbm{R}_+}\int_{\mathbbm{R}} d(\hat{R}_{GMV}, \hat{s}, t_1, f) dt_1 df \propto \exp\Bigg\{-\frac{n}{2}\frac{\left(\hat{R}_{GMV}-R_{GMV}\right)^2}{V_{GMV}}\Bigg\}
 \exp\left\{ \frac{n^2\hat{s} \left(\hat{R}_{GMV}-R_{GMV}\right)^2}{2V_{GMV}(n\hat{s}+(n-1)) } \right\} \nonumber\\
 &\times \int_{\mathbbm{R}_+} \left( \frac{f}{\hat{s}} \right)^{\frac{n-p+2}{2} + 1}\frac{e^{-\frac{f}{2\hat{s}}} }{f}d_f(f) \int_{\mathbbm{R}} e^{- \frac{((n-1) + n\hat{s})f }{2\hat{s}(n-p+1)}\left(t_1 - \frac{\hat{s}\sqrt{n-p+1}\left(\hat{R}_{GMV}-R_{GMV}\right)}{ \sqrt{V_{GMV}f}(\hat{s}-1+1/n) }\right)^2 } dt_1 df \nonumber \\
 &\propto  \left(1+\frac{n}{n-1}\hat{s}\right)^{-1/2} \exp\Bigg\{-\frac{n}{2}\frac{\left(\hat{R}_{GMV}-R_{GMV}\right)^2}{(1+\frac{n}{n-1}\hat{s})V_{GMV}}\Bigg\} \label{eqn:Rnorm} \\
 &\int_{\mathbbm{R}_+} \left( \frac{f}{\hat{s}} \right)^{\frac{n-p+1}{2}+1}\frac{e^{-\frac{(n-1)f}{2\hat{s}}} }{f}d_f(f)df. \label{eqn:shatf}
\end{flalign}}
where \eqref{eqn:Rnorm} determines the conditional distribution of $\hat{R}_{GMV}$ given $\hat{s}$ which is a normal distribution with mean $R_{GMV}$ and variance $\left(1+\frac{n}{n-1} \hat{s} \right) \frac{V_{GMV}}{n}$. The expression in \eqref{eqn:shatf} specifies the marginal distribution of $\hat{s}$ which is the integral representation of the density of the ratio of two independent variables $f$ and $\zeta$ with $(n-1)\zeta \sim \chi^2_{n-p+1}$ and $nf \sim \chi^2_{p-1} (ns)$ (c.f., \citet[Theorem 5.1.3]{mathai1992quadratic}). Hence, $n(n-p+1) / ( (n-1)(p-1)) \hat{s}$ has a noncentral $F$-distribution with $(p-1)$ and $(n-p+1)$ degrees of freedom and noncentrality parameter $ns$.
\end{proof}


\begin{proof}[Proof of Theorem~\ref{th5}]
If $\xi \sim \chi^2_{m,\delta}$, then it holds that (see, e.g., \citet[Lemma 3]{Bodnar2016125})
\begin{equation}\label{ap_th5_eq1}
\left(\frac{\xi}{m}-1-\frac{\delta}{m}\right) \stackrel{a.s.}{\rightarrow} 0\; \text{ and } \sqrt{m}\left(2\left(1+2\frac{\delta}{m} \right)\right)^{-1/2}\left(\frac{\xi}{m}-1-\frac{\delta}{m}\right) \stackrel{d}{\rightarrow} N\left(0,1 \right)
\end{equation}
for $m \rightarrow \infty$.

Throughout the proof of the theorem the asymptotic results are derived under the high-dimensional asymptotic regime, that is under $p/n \rightarrow c \in [0,1)$ as $n \rightarrow \infty$. The applications of Slutsky's lemma (c.f., \citet[Theorem 1.5]{dasgupta2008}) and Theorem~\ref{th2}, and the fact that a t-distribution with increasing degrees of freedom tends to the standard normal distribution yield the following results:

\begin{enumerate}[(i)]
 \item The application of Theorem \ref{th2}.i and \eqref{ap_th5_eq1} with $m=n-p$ leads to
 $$\sqrt{n-p}\left(\hat{V}_{GMV} - \frac{1-p/n}{1-1/n} V_{GMV} \right) \stackrel{d}{=} \frac{1-p/n}{1-1/n} V_{GMV} \sqrt{n-p}\left(\frac{\xi_2}{n-p} - 1 \right) \stackrel{d}{\rightarrow} \sqrt{2}(1-c) V_{GMV} u_1,$$
 where $u_1\sim N(0,1)$.
 \item Using \eqref{ap_th5_eq1} with $m=p-k-1$ and $\delta=n\bmu^\top \bA \bmu$, we get
 \begin{eqnarray}
   f &\stackrel{d}{=}& \frac{\xi_3}{n} + \left(s\bet + \frac{\bz_2}{\sqrt{n}} \right)^\top (\bL \bQ \bL^\top)^{-1}\left(s\bet + \frac{\bz_2}{\sqrt{n}} \right)\nonumber\\
  & =& \frac{(p-k-1)}{n} \left(\frac{\xi_3}{p-k-1} -1 - \frac{n\bmu^\top \bA \bmu}{p-k-1} \right) + \frac{(p-k-1)}{n} + \bmu^\top \bQ \bmu \nonumber \\
    &+& \frac{1}{\sqrt{n}}\left( 2s\bet (\bL \bQ \bL^\top)^{-1}\bz_2 + \frac{1}{\sqrt{n}}\bz_2^\top (\bL \bQ \bL^\top)^{-1} \bz_2 \right) \stackrel{a.s.}{\rightarrow} s+c \label{f_const}
 \end{eqnarray}
 and, hence, $
    \sqrt{n-p}\left(f - (s+p/n) \right)  \stackrel{d}{\rightarrow}  \sqrt{2(1-c)\left(c+2\bmu^\top \bA \bmu\right)}u_2 + 2s\sqrt{(1-c)} \bet^\top (\bL\bQ\bL^\top)^{-1/2} \bu_3,$
    where $u_2 \sim N(0,1)$ and $\bu_3 \sim N_k(\mathbf{0}, \bI_k)$ which are independent of $u_1$ following Theorem \ref{th2}. Furthermore, the application of \eqref{f_const} yields
    {\footnotesize
 \begin{eqnarray*}
  \sqrt{n-p}\left( \hat{R}_{GMV}-R_{GMV}\right) &\stackrel{d}{=} &\sqrt{V_{GMV}} \left( \sqrt{1-p/n} z_1 + \left(\frac{1-p/n}{1-p/n+1/n}\right)^{1/2} \sqrt{f} t_1 \right) \stackrel{d}{\rightarrow} \sqrt{V}_{GMV}\left(\sqrt{1-c} u_4 + \sqrt{s+c} u_5 \right)
 \end{eqnarray*}}
 where $u_4$, $u_5 \sim N(0,1)$ and $u_1$, $u_2$, $\bu_3$, $u_4$, $u_5$ independent.
\item Furthermore, by the stochastic representation of $\hbtheta$ as given in Theorem \ref{th2}.iii we have that in distribution $\sqrt{n-p}\left( \hat{\btheta}-\btheta\right) $
  {\footnotesize
 \begin{eqnarray*}
 &\stackrel{d}{=}& \sqrt{V_{GMV}}\Bigg(\frac{s\bet+\bz_2/\sqrt{n}}{\sqrt{f}} \sqrt{\frac{1-p/n}{1-p/n+1/n}} t_1+\left(\bL \bQ \bL^\top-\frac{\left(s\bet+\bz_2/\sqrt{n}\right)\left(s\bet+\bz_2/\sqrt{n}\right)^\top}
{f}\right)^{1/2}
\sqrt{1+\frac{t_1^2}{n-p+1}}\frac{\sqrt{n-p}}{\sqrt{n-p+2}}\bt_2\Bigg)\nonumber \\
&\stackrel{d}{\rightarrow}& \sqrt{V_{GMV}}\left(\frac{s \bet}{\sqrt{s+c}}u_5 + \left(\bL \bQ \bL^\top - \frac{s^2}{s+c}\bet \bet^\top\right)^{1/2}\bu_6 \right),
 \end{eqnarray*}}
 where $\bu_6 \sim N_k(\mathbf{0}, \bI_k)$ and is independent of $u_1$, $u_2$, $\bu_3$, $u_4$, and $u_5$.

\item The application of Theorem \ref{th2}.iv and \eqref{ap_th5_eq1} leads to
  {\small
\begin{eqnarray*}
 && \sqrt{n-p}\left( \hat{s} - \frac{(s+p/n)(1-1/n)}{1-p/n+2/n} \right) \stackrel{d}{=} \\
  &&\frac{1-1/n}{1-p/n+2/n}
\Bigg( \left(1+\frac{t_1^2}{n-p+1}\right)\frac{\sqrt{n-p}(f-(s+p/n))}{\xi_2/(n-p+2)} + (s+p/n)\left(\frac{\frac{t_1^2}{n-p+1} - \left(\frac{\xi_2}{n-p+2}-1\right)}{\xi_2/(n-p+2)}\right) \Bigg) \nonumber \\
& \stackrel{d}{\rightarrow} &\frac{1}{1-c}\Bigg( \sqrt{2(1-c)\left(c+2\bmu^\top \bA \bmu \right)}u_2 + 2s\sqrt{(1-c)} \bet^\top (\bL\bQ\bL^\top)^{-1/2} \bu_3 + \sqrt{2}(s+c)u_7\Bigg),
\end{eqnarray*}
}
where $u_7\sim N(0,1)$ and is independent of $u_1$, $u_2$, $\bu_3$, $u_4$, $u_5$, and $\bu_6$.
\item Similarly, from Theorem \ref{th2}.v we get
  {\small
\begin{eqnarray*}
&&\sqrt{n-p}\left(\hat{\bet}- \frac{s}{s+p/n}\bet \right)
   \stackrel{d}{=} \frac{1}{f} \left( \frac{-s}{s+p/n} \sqrt{n-p}\left(f-(s+p/n)\right)\bet + \sqrt{1-p/n}\bz_2 \right)\\
  &+&\frac{1}{\sqrt{f}}\left(\bL \bQ \bL^\top-\frac{\left(s\bet+\bz_2/\sqrt{n}\right)\left(s\bet+\bz_2/\sqrt{n}\right)^\top}{f}\right)^{1/2} \Bigg(\frac{1}{\sqrt{1+\frac{t_1^2}{n-p+1}}}\frac{\bt_2}{\sqrt{n-p+2}}\left(\frac{n-p}{n-p+1}\right)^{1/2}  t_1\\
&+&\left(\bI_k+f\frac{\bt_2\bt_2^\top}{n-p+2}\right)^{1/2}\left(\frac{n-p}{n-p+3}\right)^{1/2} \bt_3 \Bigg)
\stackrel{d}{\rightarrow} \frac{1}{\sqrt{s+c}}\left(\bL \bQ \bL^\top - \frac{s^2\bet \bet^\top}{s+c}\right)^{1/2}\bar{\bu}_8  \nonumber \\
&+& \frac{\sqrt{1-c}}{(s+c)} \Bigg(\bL \bQ \bL^\top - 2\frac{s^2\bet \bet^\top}{s+c} \Bigg) (\bL \bQ \bL^\top)^{-1/2} \bu_3
- \frac{s\sqrt{2(1-c)\left(c+2\bmu^\top \bA \bmu \right)}u_2}{(s+c)^2}\bet,
\end{eqnarray*}
}
where $\bu_8 \sim N_k(\mathbf{0}, \bI_k)$ and $u_1$, $u_2$, $\bu_3$, $u_4$, $u_5$, $\bu_6$, $u_7$, $u_8$ are mutually independent distributed.
\end{enumerate}
\end{proof}

\begin{proof}[Proof of Theorem~\ref{th6}:]
 The application of Theorem~\ref{th5} and of the continuous mapping theorem (c.f., \citet[Theorem 1.14]{dasgupta2008}) leads to
 $
 \bL \hat{\bw}_g  \stackrel{a.s.}{\rightarrow} \btheta + \frac{s g(R_{GMV}, (1-c)V_{GMV}, (s+c)/(1-c)) }{s+c} \bet
 $
 for $p/n \rightarrow c$ as $n \rightarrow \infty$.

Let $\hat\blambda$ and $\blambda$  be defined as in \eqref{lam}. Then, the first order Taylor series expansion yields
 \begin{flalign}
  &\sqrt{n-p}\left( \bL \hat{\bw}_g - \left(\btheta + \frac{s g\left(\blambda\right)}{s+p/n} \bet \right) \right)
  = \sqrt{n-p}\left(\hbtheta-\btheta \right) + \sqrt{n-p}\left(\hbet- \frac{s}{s+p/n}\bet\right)g\left(\hat\blambda\right) \nonumber\\ + &
  \sqrt{n-p}\left(g\left(\hat\blambda\right)- g\left(\blambda\right)\right)\frac{s\bet }{s+p/n}=  \sqrt{n-p}\left(\hbtheta-\btheta \right) + \sqrt{n-p}\left(\hbet- \frac{s}{s+p/n}\bet\right)g\left(\hat\blambda\right) \nonumber\\
  +&
  \sqrt{n-p}\begin{pmatrix}
  \hat{R}_{GMV} - R_{GMV} \\
  \hat{V}_{GMV} - (1-p/n)V_{GMV} \\
   \hat{s} - \frac{s+p/n}{1-p/n}
  \end{pmatrix}^\top
  \begin{pmatrix}
  g_1\left(\blambda\right) \\
  g_2\left(\blambda\right) \\
  g_3\left(\blambda\right)
  \end{pmatrix}
  \frac{s}{s+p/n}\bet+ o_P(1) \label{eqn:th42}
  \end{flalign}
  Hence, from Theorem \ref{th5} we get
  {\small
  \begin{flalign*}
  &\sqrt{n-p}\left( \bL \hat{\bw}_g - \left(\btheta + \frac{s g\left(\blambda\right)}{s+p/n} \bet \right) \right)
  \stackrel{d}{\rightarrow}  \sqrt{V_{GMV}}\left(\frac{s u_5 }{\sqrt{s+c}}\bet+ \left(\bL \bQ \bL^\top - \frac{s^2}{s+c}\bet \bet^\top\right)^{1/2}\bu_6 \right) \\
    +& \Bigg(\frac{1}{\sqrt{s+c}}\left(\bL \bQ \bL^\top - \frac{s^2}{s+c}\bet \bet^\top\right)^{1/2}\bu_8 + \frac{\sqrt{1-c}}{(s+c)} \Bigg(\bL \bQ \bL^\top - 2\frac{s^2}{s+c}\bet \bet^\top \Bigg) (\bL \bQ \bL^\top)^{-1/2} \bu_3 \\
  -&
    \frac{s\sqrt{2(1-c)\left(c+2\bmu^\top \bA \bmu \right)}u_2}{(s+c)^2}\bet\Bigg)g\left(\blambda\right) + g_1\left(\blambda\right)\left(\sqrt{V_{GMV}}\left(\sqrt{1-c} u_4 + \sqrt{s+c} u_5 \right) \right)\frac{s}{s+c}\bet \\
  +&
    g_2\left(\blambda\right)\left(\sqrt{2}(1-c) V_{GMV} u_1 \right)\frac{s}{s+c}\bet + g_3\left(\blambda\right)
    \Bigg(\frac{1}{1-c}\Bigg( \sqrt{2(1-c)\left(c+2\bmu^\top \bA \bmu \right)}u_2 \\
  +&
    2s\sqrt{(1-c)} \bet^\top (\bL\bQ\bL^\top)^{-1/2} \bu_3 + \sqrt{2}(s+c)u_7\Bigg) \Bigg)\frac{s}{s+c}\bet \\
  =& 
    g_2\left(\blambda\right)\frac{\sqrt{2}(1-c)V_{GMV}s}{s+c}\bet u_1
  +
    \left(\frac{g_3\left(\blambda\right)}{1-c}-\frac{g\left(\blambda\right)}{(s+c)} \right)\frac{\sqrt{2(1-c)\left(c+2\bmu^\top \bA \bmu \right)}s}{s+c} \bet u_2 \\
  +& 
    \frac{\sqrt{1-c}}{s+c}
    \left(g\left(\blambda\right)\bL \bQ \bL^\top +2s^2\left(\frac{g_3\left(\blambda\right)}{1-c} - \frac{g\left(\blambda\right)}{(s+c)}
            \right) \bet\bet^\top\right)(\bL\bQ\bL)^{-1/2}\bu_3 \\
  +& 
  \left(\frac{s\sqrt{V_{GMV}}\sqrt{1-c}}{s+c}g_1\left(\blambda\right)\right)\bet u_4
  + \left(g_1\left(\blambda\right) + 1\right)\frac{s\sqrt{V_{GMV}}}{\sqrt{s+c}}\bet u_5 \\
  +& 
 \sqrt{V_{GMV}} \left(\bL \bQ \bL^\top - \frac{s^2}{s+c}\bet \bet^\top\right)^{1/2}\bu_6 + \left(\sqrt{2}\frac{s}{1-c}g_3\left(\blambda\right)\right)\bet u_7
  + \frac{g\left(\blambda\right)}{\sqrt{s+c}}\left(\bL \bQ \bL^\top - \frac{s^2}{s+c}\bet \bet^\top\right)^{1/2}\bu_8.
\end{flalign*}
}
Using that $u_1$, $u_2$, $\bu_3$, $u_4$, $u_5$, $\bu_6$, $u_7$, $u_8$ are mutually independent and standard (multivariate) normally distributed, the expression of the asymptotic covariance matrix of $\bL\hat{\bw}_q$ is obtained.
\end{proof}

\begin{proof}[Proof of Theorem~\ref{th8}:]
  Using \eqref{hV_c}-\eqref{hbet_c} together with a first order Taylor expansion we get that
  {\small
\begin{eqnarray*}
&&\sqrt{n-p}\left(\bL\hbw_{g;c} - \bL\bw_{g}  \right)
\stackrel{d}{=} \sqrt{n-p}\left( \hbtheta - \btheta \right) - \sqrt{n-p}(\hat{s}_c-s)\frac{p/n}{\hat{s}_c(p/n+s)} g \left( \hat{R}_{GMV;c}, \hat{V}_{GMV;c}, \hat{s}_c \right) \bet \\
  &+& \sqrt{n-p} \left( \hbet - \frac{s}{s+p/n} \bet \right) \frac{\hat{s}_c + p/n}{\hat{s}_c }g \left( \hat{R}_{GMV;c}, \hat{V}_{GMV;c}, \hat{s}_c \right)\\
  &+& \sqrt{n-p} \begin{pmatrix} \hat{R}_{GMV;c} - R_{GMV} \\ \hat{V}_{GMV;c} - V_{GMV} \\ \hat{s}_c-s \end{pmatrix}^\top
		\begin{pmatrix} g_1\left(R_{GMV}, V_{GMV}, s\right) \\ g_2\left(R_{GMV}, V_{GMV}, s\right)  \\ g_3\left(R_{GMV}, V_{GMV}, s\right) \end{pmatrix}\bet  + o_P(1)\\
&=&  \sqrt{n-p}\left( \hbtheta - \btheta \right) + \sqrt{n-p} \left( \hbet - \frac{s}{s+p/n} \bet \right) \frac{\hat{s}_c + p/n}{\hat{s}_c }g \left( \hat{R}_{GMV;c}, \hat{V}_{GMV;c}, \hat{s}_c \right) +  \sqrt{n-p}\begin{pmatrix}
  \hat{R}_{GMV} - R_{GMV} \\
  \hat{V}_{GMV} - (1-p/n)V_{GMV} \\
   \hat{s} - \frac{s+p/n}{1-p/n}
  \end{pmatrix}^\top\\
  &\times&
\begin{pmatrix} g_1\left(R_{GMV}, V_{GMV}, s\right) \\ \left(1-p/n\right)^{-1} g_2\left(R_{GMV}, V_{GMV}, s\right)  \\
 \left(1-p/n\right)\left( g_3\left(R_{GMV}, V_{GMV}, s\right) -\frac{p/n}{\hat{s}_c(p/n+s)} g \left( \hat{R}_{GMV;c}, \hat{V}_{GMV;c}, \hat{s}_c \right)\right) \end{pmatrix}\bet + o_P(1)
\end{eqnarray*}
}
The rest of the proof of part (a) follows from the proof of Theorem \ref{th6}. Similarly, the statement of part (b) is obtained.
\end{proof}

{\small
\bibliography{ref}
}
\end{document}